%% file: Lora_perfs_CTM.tex
\documentclass[journal]{IEEEtran}
\IEEEoverridecommandlockouts
\usepackage{cite}
\usepackage{amsmath,amssymb,amsfonts}
\usepackage{algorithmic}
\usepackage{graphicx}
\usepackage{textcomp}

\usepackage[normalem]{ulem}
\newcommand{\stkout}[1]{\ifmmode\text{\sout{\ensuremath{#1}}}\else\sout{#1}\fi}

\usepackage{tikz}
\usepackage{xcolor,cancel}
\def\BibTeX{{\rm B\kern-.05em{\sc i\kern-.025em b}\kern-.08em
    T\kern-.1667em\lower.7ex\hbox{E}\kern-.125emX}}

\usepackage{stfloats}

\usepackage{amsthm}
\newcommand{\Mod}[1]{\ \mathrm{mod}\ #1}
\newtheorem{prop}{Proposition}
\newcommand\bydef{\stackrel{{\rm def}}{=}}

\tikzset{every picture/.style={line width=0.75pt}}

\def\PP{{\mathbb P}}

\DeclareMathOperator*{\argmax}{arg\,max}

\newcounter{mytempeqncnt}

\begin{document}

\title{Theoretical Performance of LoRa System in Multi-Path and Interference Channels}


\author{Clément~Demeslay,~\IEEEmembership{Student Member,~IEEE,}
        Philippe~Rostaing,~\IEEEmembership{Member,~IEEE,}
        and~Roland~Gautier,~\IEEEmembership{Member,~IEEE}
\IEEEcompsocitemizethanks{\IEEEcompsocthanksitem C. Demeslay, P. Rostaing and R. Gautier are with CNRS UMR 6285, Lab-STICC, from University of Brest, CS 93837, 6 avenue Le Gorgeu, 29238 Brest Cedex 3, France.\protect\\
E-mail: clement.demeslay@univ-brest.fr, philippe.rostaing@univ-brest.fr, roland.gautier@univ-brest.fr}
}

\makeatletter
\def\ps@IEEEtitlepagestyle{%
  \def\@oddfoot{\mycopyrightnotice}%
  \def\@oddhead{\hbox{}\@IEEEheaderstyle\leftmark\hfil\thepage}\relax
  \def\@evenhead{\@IEEEheaderstyle\thepage\hfil\leftmark\hbox{}}\relax
  \def\@evenfoot{}%
}
\def\mycopyrightnotice{%
  \begin{minipage}{\textwidth}
  \centering \scriptsize
  Copyright~\copyright~2021 IEEE. Personal use of this material is permitted. However, permission to use this material for any other purposes must be obtained from the IEEE by sending a request to pubs-permissions@ieee.org.
  \end{minipage}
}
\makeatother

\maketitle

\begin{abstract}
This paper presents a semi-analytical approximation of Symbol Error Rate ($SER$) for the well known LoRa Internet of Things ($IoT$) modulation scheme in the following two scenarios: 1) in multi-path frequency selective fading channel with Additive White Gaussian Noise ($AWGN$) and 2) in the presence of a second interfering LoRa user in flat-fading $AWGN$ channel.
Performances for both coherent and non-coherent cases are derived by considering the common Discrete Fourier transform ($DFT$) based detector on  the received LoRa waveform.
By considering these two scenarios, the detector exhibits parasitic peaks that severely degrade the performance of the LoRa receiver.
We propose in that sense a theoretical expression for this result, from which a unified framework based on peak detection probabilities allows us to derive $SER$, which is validated by Monte Carlo simulations.
Fast computation of the derived closed-form $SER$ allows to carry out deep performance analysis for these two scenarios.
\end{abstract}

\begin{IEEEkeywords}
LoRa waveform, chirp modulation, multi-path channel, LoRa interference,  performance analysis.
\end{IEEEkeywords}

\section{Introduction}

The Internet of Things ($IoT$) is experiencing striking growth since the past few years enabling much more devices to communicate and allowing many scenarios to be a reality such as smart cities.
The number of $IoT$ devices is expected to rapidly grow, jumping from almost 10 to more than 21 billion \cite{IoTDevicesNumber}.
Many technologies were developed in that sense relying on licensed bands (Narrow Band $IoT$ ($NB-IoT$), Extended Coverage GSM (EC-GSM) and LTE-Machine (LTE-M)) or unlicensed bands such as SigFox, Ingenu, Weightless or Long Range (LoRa) \cite{goursaud}.
We will focus on LoRa in this paper.
LoRa was initially developed by the French company Cycleo in 2012 \cite{lora_patent} and is now the property of Semtech company, the founder of LoRa Alliance. 
LoRa is nowadays a front runner of LP-WAN solutions and holds a lot of attention by the scientific research community. 
Due to its patented nature, initial research was mainly based on retro-engineering of existing LoRa transceivers \cite{knight_gnuradio}. 
The first paper to provide a rigorous mathematical representation of LoRa signals and its demodulation scheme was achieved by \cite{vangelista}.
Further research were conducted focusing on LoRa network capacity enhancements \cite{joerg3}, channel coding improvements \cite{joerg4},\cite{imt} or temporal and frequency synchronization techniques \cite{xhonneux},\cite{bernier}.

There was previous research work addressing Multi-Path Channel ($MPC$) impact on performance with an experimental point of view (see Section \ref{subsec:RelatedWorkMPC}) but to the best of our knowledge, theoretical assessment of a such effect has not been investigated yet in the literature.
Even if $MPC$ may produce small Inter Symbol Interference ($ISI$) depth for the detection of the current symbol, the coherent and non-coherent detectors are very sensitive to one or several significant echoes.
This situation may be encountered in outdoor environments (see Section \ref{LoRaChannelEffect} for more details).
In this paper, a tight approximation of $SER$ in $MPC$ is proposed for  coherent and non-coherent LoRa detection schemes.
Performance degradation of LoRa modulation is then studied for the two-path and the exponential decay channel models.

We  show  that  additional  interference peaks appear by using the common Discrete Fourier  Transform ($DFT$)-based detector on the received dechirping LoRa waveform over $MPC$. 
$SER$ derivation can  be seen as a problem of correct peak detection at the $DFT$-output against interference and noise peaks in presence of $AWGN$.
$SER$ is derived in two steps as follows. 
First, close approximations are provided to obtain closed-form expression, in terms of Cumulative Distribution Functions ($CDF$s), of the detection probability of the correct peak for a given noise random sample at the corresponding $DFT$-output.
Similar developments have been done for the flat-fading $AWGN$ channel \cite{ferre1} but we extend here results for $MPC$. 
And secondly, expectation of the closed-form expression over complex Gaussian random variable is computed numerically by using two-dimensional Cartesian products of one-dimensional Gauss-Hermite quadrature formulae \cite{Abramowitz1972}.
The accuracy of the derived approximations is then confirmed by comparisons to numerical results.
This procedure provides a fast and tight $SER$ estimation which allows to carry out deep performance analysis of LoRa modulation in $MPC$.

From the same unified framework, performance evaluation in terms of $SER$ is derived in case of LoRa interference, where the desired LoRa signal is corrupted by interfering LoRa signal in $AWGN$ flat-fading channel. 
A fine analysis of performance in relation to path delay and complex path gain of the interferer channel is also addressed.
We compare our derived $SER$ expression with previous work in \cite{joerg5} about performance of LoRa interference where the  $SER$ expression is derived by a different approach.

\textcolor{black}{
The novelty and contributions of the paper can be highlighted in the following:
\begin{itemize}
    \item Proposing a unified framework to derive semi-analytical $SER$ approximation for both discrete-time $MPC$ and aligned LoRa interference with the same Spreading Factor ($SF$) parameter. In a first approximation, we supposed the channel path delays (or the time delay of the interferer LoRa signal) as multiple of sampling rate.
    \item Analyzing LoRa performance for coherent and non-coherent receivers for the one-tap, and the exponential decay, discrete channel models.
    \item Assessing LoRa performance for the non-coherent receiver in presence of LoRa interference.
Performance analysis is studied as a function of LoRa interferer path delays, it confirms results on previous research work \cite{joerg5} and adds some noticeable refinements for specific LoRa interferer delays around the half of the symbol period.
\end{itemize}}

The remainder of the article is organized as follows.
In Section \ref{LoRaOverview}, we review the basics of LoRa physical layer. 
In Section \ref{LoRaPartIII}, we introduce the $MPC$ model and its application to LoRa wave-forms. 
The semi-analytical  \textcolor{black}{$SER$} expression is then derived under several close approximations in Section \ref{sec::SEP}.
Section \ref{sec:lora_interf} introduces the LoRa interference model, the associated theoretical  \textcolor{black}{$SER$} and the study of performance impact of the relative phase difference between the signal of interest and the interferer.
Finally, in Section \ref{sec:Simu}, simulation results confirm theoretical derivations for both $MPC$ and LoRa interference models.

\textcolor{black}{
\section{Related work}
There was many studies in the literature focusing on $MPC$ and LoRa interference impact that were possible thanks to preliminary LoRa research dealing with theoretical performance of original non-coherent LoRa system.
We can cite for example the authors in \cite{joerg2} who derived a closed-form approximation of LoRa $BER$ performance in $AWGN$ and Rayleigh channels.
Their work was further extended to the case of coded LoRa communications \cite{burg} in which interleaving, Gray and Hamming coding are considered.
They evaluate performance with residual Carrier Frequency Offset ($CFO$) and show that the latter has strong impact on performance and therefore proper $CFO$ mitigation techniques must be implemented.
\subsection{$MPC$}\label{subsec:RelatedWorkMPC}
The impact of $MPC$ channels was investigated with the following studies.
The authors in \cite{imt2} evaluated the impact of time/frequency selective channels on demodulation process and highlighted the good LoRa resiliency on these channels and especially with the latter.
Furthermore, an improved LoRa detector based on cyclic cross-correlation to combat $MPC$ was also proposed in \cite{guo}.
Experimental approach was proposed in \cite{staniec} where the authors assessed the impact of both electromagnetic interference and heavy $MPC$ in anechoic and reverberation chambers.
They came with the conclusion that LoRa is robust to  their experimental Rayleigh fading  $MPC$ only for $SF \ge 10$  and whatever the signal bandwidth.
\subsection{LoRa interference}
The LoRa interference case was also considered with many studies evaluating theoretical performance in same $SF$ with the interference signal delayed by an integer number of sampling periods (aligned) in \cite{joerg5}.
It is shown in this study that a such interference has dramatic impact on performance, leading to a sensitivity threshold approaching $\infty$ if the 
signal-to-interference ratio ($SIR$) approaches $SIR_{dB}=0$.
The authors also highlights good LoRa resiliency, with interference free performance recovered with $SIR$ greater than $10$ dB.
Non-integer interference delay (non-aligned) has been also investigated in \cite{afisiadis},\cite{afisiadis2} and the latter was recently extended in \cite{afisiadis3} to the coherent receiver and including hardware impairments such as $CFO$.
Algorithms were also designed to enable the decoding of a LoRa symbol stream contaminated by a single or multiple LoRa user \cite{rachkidy},\cite{laporte}.
}

\section{LoRa modulation overview}\label{LoRaOverview}

\subsection{LoRa wave-forms}\label{LoRaSignals}

In the literature, LoRa wave-forms are of the type of Chirp Spread Spectrum ($CSS$) signals.
These signals rely on sine waves with Instantaneous Frequency ($IF$) that varies linearly with time over frequency range $f \in [-B/2,B/2]$ ($B\in\{125,250,500\}$~kHz) and time range $t \in [0,T]$ ($T$ the symbol period).
This basic signal is called an \emph{up-chirp} or \emph{down-chirp} when frequency respectively increases or decreases over time.
A LoRa symbol consists of $SF$ bits ($SF \in \{7,8,\ldots,12\}$) leading to an $M$-ary modulation with $M = 2^{SF} \in \{128,256,\ldots,4096\}$.
In the discrete-time signal model, the Nyquist sampling rate ($F_s=1/T_s$) is usually used \textit{i.e.} $T_s = 1/B = T/M$.
The signal symbol has then $M$ samples.
Each symbol $a \in \{0,1,\ldots,M-1\}$ is mapped to an \emph{up-chirp} that is temporally shifted by $\tau_a = a T_s$ period.
We may notice that a temporal shift $\tau_a=aT_s$ conducts to shift by $aB/M=a/(MT_s)=a/T$ the $IF$.
The modulo operation is applied to ensure that $IF$ remains in the interval $[-B/2,B/2]$.
This behavior is the heart of $CSS$ process.
A mathematical expression of LoRa wave-form sampled at $t = k T_s$ has been derived in\cite{chiani} :

\begin{equation}
    x(kT_s;a) \triangleq x_a[k] = e^{2j\pi k \left( \frac{a}{M} - \frac{1}{2} + \frac{k}{2M} \right)} \quad k = 0,1,\ldots,M-1
    \label{eqxa}
\end{equation}

We may see that an \emph{up-chirp} is actually a LoRa wave-form with symbol index $a = 0$, written $x_0[k]$.
Its conjugate $x_0^*[k]$ is then a \emph{down-chirp}. 

\subsection{LoRa demodulation scheme}\label{LoRaDemodScheme}

Reference \cite{vangelista} proposed a simple and efficient solution to demodulate LoRa signals.
In $AWGN$ channel, the demodulation process is based on the Maximum Likelihood ($ML$) detection scheme.
The received signal is:
\begin{equation}
    r[k] = x_a[k] + w[k]
\end{equation}
with $w[k]$ a complex $AWGN$ with zero-mean and variance $\sigma^2 = E[|w[k]|^2]$.
$ML$ detector aims to select index $\widehat{a}$ that maximises the scalar product $\langle r[k],x_n[k] \rangle$ for $n \in \{0,1,\ldots,M-1\}$ defined as:
\begin{equation}
\begin{split}
    \langle r[k],x_n[k] \rangle &= \sum_{k=0}^{M-1} r[k] x_n^*[k]\\
    &= \sum_{k=0}^{M-1} \underbrace{r[k]x_0^*[k]}_{\Tilde{r}[k]} e^{-j2\pi\frac{n}{M}k}= \Tilde{R}[n].
\end{split}
\end{equation}

The demodulation stage proceeds with two simple operations:
\begin{itemize}
    \item multiply the received signal by the \emph{down-chirp} $x_0^*[k]$, also called dechirping,
    \item compute $\Tilde{R}[n]$, the $DFT$ of $\Tilde{r}[k]$ and select the discrete frequency index $\widehat{a}$ that maximizes  $\Tilde{R}[n]$.
    \label{eqDemodLoRa}
\end{itemize}

This way, the dechirp process merges all the signal energy in a unique frequency bin $a$ and can be easily retrieved by taking the magnitude (non-coherent detection) or the real part (coherent detection) of $\Tilde{R}[n]$.
The symbol detection is then:

\begin{equation}
    \begin{split}
        \widehat{a}_{NCOH} &= \underset{n}{\argmax} \quad |\Tilde{R}[n]|^2 \equiv \underset{n}{\argmax} \quad |\Tilde{R}[n]|
    \end{split}
    \label{eqDemodLoRaNCOH}
\end{equation}

for non-coherent detection, and:

\begin{equation}
     \widehat{a}_{COH} = \underset{n}{\argmax} \quad \Re \{ \Tilde{R}[n] \}
    \label{eqDemodLoRaCOH}
\end{equation}

for coherent detection with $\Re\{x\}$ denoting the real part of complex number $x$.

\section{Multi-path channel on LoRa signal}\label{LoRaPartIII}
\subsection{Multi-path channel model}\label{LoRaChannelEffect}

We study in this section the effect of $MPC$ on LoRa signals.
The discrete-time channel model used is as follows:

\begin{equation}
     c[k] = \sum_{i=0}^{K-1} \alpha_i \delta[k - k_i]
    \label{eqChannelModel}
\end{equation}

with $K$ the number of paths and $\alpha_i=|\alpha_i|e^{j\phi_i}$ the complex path gain arriving at tap $k_i$.
A sufficient condition to consider a channel as frequency selective is $k_i \ge 1$ \textit{e.g.} $B = 500$~kHz, $T_s = 1/B = 2~\mu$s.
This value is a typical path delay seen in outdoor environments (few $\mu$s usually).
Indeed, the COST 207 channel model, originally developed for Global System for Mobile Communications ($GSM$)
in \cite{cost_207}, propose typical channel tap settings for various situations such as difficult hilly urban environments, denoted as Bad Urban ($BU$) channel.
LoRa devices are susceptible to be implemented in this type of environments.
Furthermore, $GSM$ bands ($GSM 900$) are close to LoRa ones (868 MHz in Europe), COST 207 channel is thus relevant in this case.
The 12-tap $BU$ channel configuration exhibits many echoes with strong magnitudes and high relative delays.
The $9^{th}$ tap has for example a delay of 6 $\mu s$ and a relative power of 0.8.
This corresponds for LoRa to have $k_i = 3$ for $B = 500$ kHz.
We expect then that the largest echo $k_{max} \ll M$, that is, an $ISI$ only between the current and previous symbol over a reduced number of samples.
The symbol detector presented herein is very sensitive to significant path delays although $ISI$ depth is small.
In this section, we evaluate the performance impact of $MPC$ on LoRa wave-forms.

We consider a set of transmitted symbols $a_l$ ($l=0,\ldots L-1$) as:
\begin{equation}
     s[k^{\prime}] = \sum_{l=0}^{L-1} x_{a_l}[k^{\prime}\!\Mod{M}]
    \label{eqsk}
\end{equation}

for $k^{\prime} = k + lM$ and $k=0,\ldots,M-1$.
The received signal is then:
\begin{equation}
     r[k^{\prime}] = \underbrace{c[k^{\prime}] * s[k^{\prime}]}_{m[k^{\prime}]} + w[k^{\prime}].
    \label{eqyk}
\end{equation}

We note $\sigma_{\Re}^2 = \sigma_{\Im}^2 = \sigma^2/2$, the variance of real and imaginary part of $w[k^\prime]$.
$m[k^{\prime}]$ is the received waveform after channel effect.

\subsection{Channel effect on LoRa wave-form}

Let us denote $r_a[k]$ the received signal plus noise for detecting the current
symbol $a$ into its symbol interval for $k=0,\ldots,M-1$.
We suppose that the receiver is synchronized on the first path (\textit{i.e.} $k_0=0$).

\begin{prop}
Performing the \emph{down-chirp} operation $x_0^*[k]$ to $r_a[k]$ yields:
\setlength\arraycolsep{2pt}
\begin{eqnarray}
\nonumber
  \Tilde{r}_a[k]&=&x_0^*[k]r_a[k] \\
&=&\alpha_0 e^{2j\pi k \frac{a}{M}} +
  \sum_{i=1}^{K-1} \Tilde{\alpha}_i(\overline{a}) e^{2j\pi k
    \frac{\overline{a} - k_i}{M}} + \Tilde{w}[k]
  \label{eqrak2_IES}
\end{eqnarray}
where $\Tilde{w}[k] \sim \mathcal{CN}(0,\sigma^2)$ and:\\
\begin{equation}
\label{eq:alpha_i_tilde}
 \Tilde{\alpha}_i(\bar{a})=\alpha_i x_{\bar{a}}[M - k_i]=\alpha_i e^{-2j\pi
   k_i \frac{\bar{a}}{M}} x_0[M - k_i]
\end{equation}
with:
\begin{equation}
\overline{a} \triangleq 
\begin{cases}
  a^- & \text{ for } k = 0,1,\ldots,k_i-1 \hfill\text{ (previous symbol)} \\
  a & \text{ for } k = k_i,\ldots,M-1 \hfill \text{ (current symbol).}
\end{cases}
\label{eq:bar_a}
\end{equation}
\end{prop}

\begin{figure}[t]
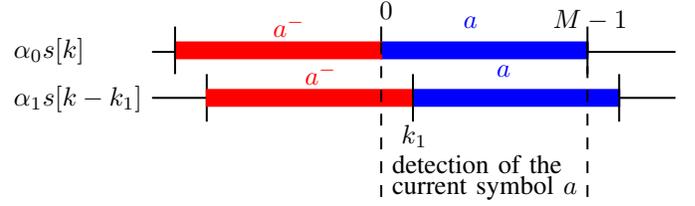

  \centering
  \input figures/demes1.tex
  \caption{Illustration of $ISI$ for detecting the current symbol $a$ in case of two-path channel at delays $k_0=0$ (synchronized on the first path) and $k_1$.}
\label{fig:ISI}
\end{figure}

\begin{proof}
For the sake of simplicity, we first consider the two-path channel.  
The received signal is then $r[k]=\alpha_0 s[k]+\alpha_1 s[k-k_1] + w[k]$.
By focusing in the detection interval $k=0,\ldots,M-1$ of the current symbol $a$ (Figure \ref{fig:ISI}), the signal on the synchronized path is equal to $s[k]=x_a[k]$ and the signal on the delayed path can be expressed as:
\begin{equation}
  s[k-k_1]=\left\{
\begin{array}{ll}
x_{a^{-}}[M-k_1+k] &\text{for } k=0,\ldots,k_1-1\\
x_{a}[k-k_1] &\text{for } k=k_1,\ldots,M-1.
\end{array}
\right.
\end{equation}

From \eqref{eqxa} one can verify the property $x_a[M-n]=x_a[-n]$ for $n=0,\ldots,M-1$, then the received signal for the detection of the current symbol $a$ could be expressed as:
\begin{equation}
r_a[k]=\alpha_0x_a[k]+\alpha_1x_{\bar{a}}[k-k_1] + w[k]
\label{eq:r_a[k]}
\end{equation}
where $\bar{a}$ is defined in (\ref{eq:bar_a}).
By substituting (\ref{eqxa}) into (\ref{eq:r_a[k]}) yields:
\begin{multline}
  r_a[k]=\alpha_0x_a[k]+\alpha_1e^{j2\pi(\frac{\bar{a}}{M}-\frac{1}{2}+\frac{n-k_1}{2M}) }
e^{-j2\pi k_1(\frac{\bar{a}}{M}-\frac{1}{2}+\frac{k-k_1}{2M})}\\ + w[k].
\end{multline}

By multiplying the \emph{down-chirp} to $r_a[k]$ we obtain after some basic manipulations:
\setlength\arraycolsep{2pt}
\begin{eqnarray}
\Tilde{r}_a[k]&=&x_0^*[k]r_a[k] \nonumber\\
&=&\alpha_0e^{j2\pi k\frac{a}{M}}+\alpha_1e^{-j2\pi
  k_1(\frac{\bar{a}}{M}-\frac{1}{2}-\frac{k_1}{2M})}e^{j2\pi
  n\frac{\bar{a}-k_1}{M}} + \Tilde{w}[k] \nonumber\\[-1em]
&& \nonumber\\
&=&\alpha_0e^{j2\pi
  k\frac{a}{M}}+\alpha_1x_{\bar{a}}[-k_1]e^{j2\pi n\frac{\bar{a}-k_1}{M}} + \Tilde{w}[k] \nonumber \\
&=&\alpha_0e^{j2\pi
  k\frac{a}{M}}+\Tilde{\alpha_1}(\bar{a})e^{j2\pi n\frac{\bar{a}-k_1}{M}} + \Tilde{w}[k]
\end{eqnarray}
with $\Tilde{\alpha_1}(\bar{a})=\alpha_1x_{\bar{a}}[-k_1]=\alpha_1x_{\bar{a}}[M-k_1]$.
By applying the same development for $K>2$ paths, the general expression is
straightforward and given in (\ref{eqrak2_IES}).
\end{proof}

\subsection{$DFT$ of the received down-chirping LoRa signal}

The second operation in the demodulation stage is to compute the $DFT$ of $\Tilde{r}_a[k]$ and select the discrete frequency index that maximizes the $DFT$ magnitude for non-coherent detection, or $DFT$ real part for coherent detection.

{\prop The $DFT$ $\Tilde{R}_a[n]$ of $\Tilde{r}_a[k]$ assuming self-$ISI$ (\textit{i.e.} $a^-=a$) is:
\begin{equation}
\Tilde{R}_a[n] = M \alpha_0 \delta[n - a] + M \sum_{i=1}^{K-1}
    \Tilde{\alpha}_i(a) \delta[n - a + k_i] + \Tilde{W}[n]
    \label{eq:TFD1}
\end{equation}
for $n = 0,1,\ldots,M-1$ and $\Tilde{W}[n] \sim \mathcal{CN}(0,M\sigma^2 = \sigma_w^2$).
}
\begin{proof}
For $a^-=a$ (the previous symbol $a^-$ is equal to the current symbol $a$) and thanks to the property $x_a[-k]=x_a[M-k]$, channel effect could be seen as a circular convolution and from \eqref{eq:r_a[k]} we obtain directly the $DFT$ given in~(\ref{eq:TFD1}).
\end{proof}

In case of self-$ISI$ ($a^-=a$), the $DFT$ output \eqref{eq:TFD1} is particularly simple and exhibits peaks at index frequencies $n=a-k_i$ for $i=0,\ldots,K-1$.
The desired symbol $a$ is located at $n=a$ ($k_0=0$) but additional interference peaks appear due to multiple echoes.
Note that even if the self-$ISI$ case occurs only with the probability $1/M$ for \textit{i.i.d.} symbols, we have to consider this case to derive theoretical performance results presented in Section~\ref{sec::SEP}.
Now, let us focusing on the more general $ISI$ case.

{\prop The $DFT$ $\Tilde{R}_a[n]$ of $\Tilde{r}_a[k]$ in presence of $ISI$ (\textit{i.e.} $a^-\neq a$) is:
\begin{equation}
\label{eqRak2_IES}
\begin{split}
     & \Tilde{R}_a[n] = M \alpha_0 \delta[n - a] \\
     & + \sum_{i=1}^{K-1} \left\{ (M - k_i) \Tilde{\alpha}_i(a) + M_i[a-k_i;a^-] \right\} \delta[n - a + k_i] \\
     & + \left\{ M_i[n;a^-] - M_i[n;a] \right\} (1 - \delta[n - a +k_i]) + \Tilde{W}[n]
     \end{split}
\end{equation}
with:
\begin{equation}
        M_i[n;\bar{a}] = \Tilde{\alpha}_i(\bar{a}) \sum_{k=0}^{k_i - 1} e^{2j\pi k/M(\bar{a} - k_i - n)}
    \label{eqMi}
\end{equation}
where $n$ is the frequency index, $\bar{a}=a$ or $\bar{a}=a^-$, and $\Tilde{\alpha}_i(\bar{a})$ is given in \eqref{eq:alpha_i_tilde}.
}
\begin{proof}
See appendix \ref{sec:outDFT}.
\end{proof}

We may note that $ISI$ reduces interference peaks even more when path delay is significant.
The peak magnitude is indeed $(M-k_i) \Tilde{\alpha}_i(a)$ (dominant term) for $n=a-k_i$ in (\ref{eqRak2_IES}) is lesser than $M \Tilde{\alpha}_i(a)$ in \eqref{eq:TFD1}.
Performance is then improved for longer path delay (\textit{e.g.} $k_1 = 11$ versus $k_1=1$).
Simulations in sections \ref{subsec:valid} and \ref{sec:degradation} will highlight this point.
In contrast with \eqref{eq:TFD1}, we observe also parasitic peaks everywhere outside the index frequencies $n=a-k_i$.
However, we will further see that these peaks vanish in presence of noise.
Note that for $a=a^-$ \eqref{eqRak2_IES} can be reduced to (\ref{eq:TFD1}).
Indeed, the third term in \eqref{eqRak2_IES} cancels out and, in the second term, $M_i[a-k_i;a^-]$ equals $\Tilde{\alpha}_i(a)k_i$.

\section{LoRa  \textcolor{black}{$SER$} under multi-path channel}\label{sec::SEP}

This part presents the derivation of  \textcolor{black}{$SER$} noted as $P_e$.
We will extend the method derived in \cite{ferre1} for the one-path channel to the case of $MPC$.

\subsection{General expressions of  \textcolor{black}{$SER$} for non-coherent and coherent detection schemes}\label{sec:globalSEP} 

As seen in previous part, peaks magnitude at $DFT$ output are different for the cases $a=a^-$ and $a\neq a^-$.
Hence different $P_e$ expressions for those two cases need to be derived.
Let us denote hypothesis $H_a$, the current symbol is $a$, and $H_{a^-}$, the previous symbol is $a^-$.
According to the law of total probability, $P_e$ is expressed as:
\begin{equation}
    P_e = \sum_{a,a^-=0}^{M-1} \mathbb{P}[\widehat{a} \ne a/H_a,H_{a^-}]\, \mathbb{P}[H_a,H_{a^-}].
    \label{eqPeBayes}
\end{equation}

By separating the terms in \eqref{eqPeBayes} for $a=a^-$ and $a \neq a^-$ and for \textit{i.i.d.} symbols (\textit{i.e.} $\mathbb{P}[H_a,H_{a^-}] = \frac{1}{M^2}$), $P_e$ leads to:
\begin{eqnarray}
        P_e &=& \frac{1}{M^2} \sum_{a=0}^{M-1} \mathbb{P}[\widehat{a} \ne a/H_a,H_{a^-=a}] \label{eqPe1} \\
        && + \frac{1}{M^2} \sum_{\substack{a,a^-=0 \\ a \neq
            a^-}}^{M-1}\mathbb{P}[\widehat{a} \neq a/H_a,H_{a^- \neq a}].
    \label{eqPe}
\end{eqnarray}

As we will see further  (in section~\ref{subsec:Pd}) for non-coherent detection scheme, $\mathbb{P}[\widehat{a}\ne a/H_a,H_{a^-=a}]=P_e^{(1)}$ does not depend on the transmitted symbol $a$ but $\mathbb{P}[\widehat{a} \ne a/H_a,H_{a^- \neq a}]={P}_e^{(2)}(a,a^-)$ depends on $a$ and $a^-$ for $a \neq a^-$.
However, we'll consider $\widehat{P}_e^{(2)}$ as a good approximation of $P_e^{(2)}(a,a^-)$ where $\widehat{P}_e^{(2)}$ does not depend on $a$ and 
$a^-$, $\forall a \neq a^-$.
Then, (\ref{eqPe1}) and (\ref{eqPe}) can be simplified as:
\begin{equation}
    \begin{split}
    P_e \simeq \frac{1}{M} P_e^{(1)} + \frac{M-1}{M} \widehat{P}_e^{(2)}.
    \end{split}
    \label{eqPeApprox}
\end{equation}

Otherwise, for the coherent detection scheme we will further see that $\mathbb{P}[\widehat{a} \ne a/H_a,H_{a^-=a}]=P_e^{(1)}(a)$ which depends on $a$, and the second term $P_e^{(2)}(a,a^-)$ could be approximated by $\widehat{P}_e^{(2)}(a)$ which depends only on $a$.  We obtain for the coherent case:
\begin{equation}
      P_e \simeq \frac{1}{M^2} \sum_{a=0}^{M-1} \left(P_e^{(1)}(a) + (M-1)\widehat{P}_e^{(2)}(a)\right).
    \label{eqPeApprox2}
\end{equation}

Unfortunately, the computational complexity of $P_e$ is $M$ times greater for the coherent case than the non-coherent case.
However, the numerical evaluation of \eqref{eqPeApprox2}, even for $M=2^{12}=4096$, doesn't make any computing difficulty.

The evaluation of $P_e$ depends on the magnitude for non-coherent detection (or the real part for coherent detection) of the DFT output
$\Tilde{R}_a[n]$ $\forall a$, with $\Tilde{R}_a[n]$ given in \eqref{eq:TFD1} and \eqref{eqRak2_IES} for $n=0,1,\ldots,M-1$.

As the  \textcolor{black}{$SER$} $P_e^{(c)}$ (for $c=1,\,2$) could be directly derived from the probability of detection of the correct symbol $a$ at the $DFT$ output, we first compute this probability for a given noise random sample $\Tilde{W}[a]$ (at the correct $a$-peak frequency index).

\subsection{Probability of detection of the correct symbol $a$ for a given noise random sample at the $a$-peak index}\label{subsec:Pd}

From \eqref{eqRak2_IES} in the $ISI$ case ($a^-\neq a$), several approximations can be made to $\Tilde{R}_a[n]$ regarding $n$ values:
\begin{itemize}
    \item {for $n = a$}
    \begin{IEEEeqnarray}{rCl}
    \Tilde{R}_a[a]&=&M\alpha_0+\underbrace{\sum_{i=1}^{K-1}(M_i[a;a^-]-M_i[a;a])}_{I\approx0}+\Tilde{W}[a]\nonumber\\[-1.5em]
    &&\IEEEyesnumber\label{eq:Ya1_ISI}\\
     \Tilde{R}_a[a]&\approx&M\alpha_0+\Tilde{W}[a],\IEEEyessubnumber\label{eq:Ya1_ISI_Approx}
    \end{IEEEeqnarray}
\item {for $n = a - k_i$} 
 \begin{IEEEeqnarray}{rCl}
    \Tilde{R}_a[a\!-\!k_i]&=&(M\!-\!k_i)\Tilde{\alpha}_i(a)\!+\!\underbrace{M_i[a\!-\!k_i;a^-]}_{\approx 0}+\Tilde{W}[a\!-\!k_i]\nonumber\\[-1em]
    &&\IEEEyesnumber\label{eq:Ya2_ISI}\\
    \Tilde{R}_a[a\!-\!k_i]&\approx&(M\!-\!k_i)\Tilde{\alpha}_i(a)+\Tilde{W}[a\!-\!k_i],
\IEEEyessubnumber\label{eq:Ya2_ISI_Approx}
  \end{IEEEeqnarray}
    \item {for $n \ne a$ and $n \ne a - k_i$}
    \begin{IEEEeqnarray}{rCl}
     \Tilde{R}_a[n]&=&\underbrace{\sum_{i=1}^{K-1}(M_i[n;a^-]-M_i[n;a])}_{J\approx0}+\Tilde{W}[n]\IEEEyesnumber\label{eq:Ya3_ISI}\\[-.5em]
       \Tilde{R}_a[n]&\approx&\Tilde{W}[n].\IEEEyessubnumber\label{eq:Ya3_ISI_Approx}
     \end{IEEEeqnarray}
\end{itemize}

Merging Eqs. \eqref{eq:Ya1_ISI_Approx}-\eqref{eq:Ya2_ISI_Approx}-\eqref{eq:Ya3_ISI_Approx} yields finally:
\begin{equation}
    \begin{split}
     & \Tilde{R}_a[n] \;\textcolor{black}{\approx}\; M \alpha_0 \delta[n - a] \\
     & + \sum_{i=1}^{K-1} (M - k_i) \Tilde{\alpha}_i(a) \delta[n - a + k_i] + \Tilde{W}[n].
     \end{split}
     \label{eq:Ya_ISI_approx}
\end{equation}

The approximation $\widehat{P}_e^{(2)} \approx P_e^{(2)}(a,a^-)$ (non-coherent case) or $\widehat{P}_e^{(2)}(a)\approx P_e^{(2)}(a,a^-)$ (coherent case) described in section~\ref{sec:globalSEP} supposes that the quantities $I,J$ and $M_i[a-k_i;a^-]$ are negligible.
These quantities depend on $M_i[\,\textbf{.}\,;a]$ and $M_i[\,\textbf{.}\,;a^-]$ that are indeed non-coherent complex exponential sums.
These quantities vanish in comparison to the noise $\Tilde{W}[n]$ because the noise standard deviation $\sqrt{M}\sigma$ is much larger, in particularly for the useful low-$SNR$ range.
Moreover, one can verify that $M\alpha_0\gg I$ in \eqref{eq:Ya1_ISI} and $(M-k_i)\Tilde{\alpha}_i(a)\gg M_i[a-k_i;a^-]$ in~\eqref{eq:Ya2_ISI}.
Note that for the special case $n=a^--k_i$ that could occur in \eqref{eq:Ya3_ISI}, $M_i[n,a^-]$ doesn't correspond to a sum of non-coherent complex exponential terms and is equal 
to $k_i \Tilde{\alpha}_i(a^-)$. However, this term remains vanish because $k_i\ll M$ (small $ISI$-depth in comparison to the symbol duration) and is small in comparison to the correct $a$-peak detection of amplitude $M$.

From \eqref{eq:TFD1} in the self-$ISI$ case ($a^-=a$), $\Tilde{R}_a[n]$ yields:
\begin{itemize}
    \item {for $n = a$}
      \begin{equation}
    \Tilde{R}_a[a] = M\alpha_0 + \Tilde{W}[a],
      \end{equation}    
    \item {for $n = a - k_i$}
      \begin{equation}
    \Tilde{R}_a[a-k_i] = M \Tilde{\alpha}_i(a) + \Tilde{W}[a-k_i],
    \label{eq:Ya2_2}
      \end{equation}
    \item {for $n \ne a$ and $n \ne a - k_i$}
      \begin{equation}
      \Tilde{R}_a[n] = \Tilde{W}[n].
      \end{equation}      
\end{itemize}

Notice that thanks to the previous simplifications in the $ISI$ case, the difference between $ISI$ and self-$ISI$ cases appears only at $n=a-k_i$ in \eqref{eq:Ya2_ISI_Approx} and \eqref{eq:Ya2_2} where the term $(M-k_i) \Tilde{\alpha}_i(a)$ must be considered for $a^-\neq a$ instead of $M \Tilde{\alpha}_i(a)$ for $a^-=a$.
The term $\Tilde{\alpha}_i(a)$ is given in \eqref{eq:alpha_i_tilde}.
A comparison between \eqref{eqRak2_IES} and \eqref{eq:Ya_ISI_approx} expressions for $ISI$ case is presented in Figure \ref{fig:lora_magnitude_mpc} without the noise term $\Tilde{W}[n]$ for two different symbols $a$ and $a^-$, and a two-path channel.
As seen in the figure, the approximated expression is very close to the exact expression.
This approximation depends however on $a^-$ and $a$ for a given $k_i$ and the deviation may thus vary.

\textcolor{black}{Figure \ref{fig:lora_magnitude_mpc_non_aligned} shows the $DFT$ output of the dechirped signal in the case of a non-aligned channel \textit{i.e.} $k_i$ a non-integer value formed by integer and fractional parts of sampling period, denoted as $L_i = \left \lfloor{k_i}\right \rceil$ and $\eta_i = k_i - L_i$, respectively.
As the receiver is synchronized on the first received path, $k_0 = 0$.
We consider in the figure a two-path channel with $\alpha_0=1$, $\alpha_1=0.7$ and different echo delay values $k_1 = \{4,10.25,20.5,30.75\}$, $SF=7$.
For figure clarity, we consider self-$ISI$ case ($a^- = a = 80$), each echoes are plotted with different colors and are clearly separated to each other.
An oversampling factor of $R=8$ is used to simulate $\eta$.
We may see that the non-aligned echoes have their $DFT$ energy bin at $n = a - L_1$ spread over neighbor bins.
The spread increases as $\eta$ grows and is maximum when $\eta=\pm0.5$.
$\delta_{k_1}$ denotes in the figure the magnitude difference between the peak of interest at $n=a=80$ and the echo at $n = a - L_1$.
The reported values are $\delta_{k_1=4} \approx 37.2, \delta_{k_1=10.25=30.75} \approx 47.5 \mbox{ and } \delta_{k_1=20.5} \approx 71.3$.
Lower $\delta_{k_1}$ values increase sensitivity of the detector to the noise.
We expect then that the aligned channel will reduce performance as the entire echo energy is contained in a single $DFT$ bin and will be more harmful for the detector.
}

\begin{figure}[ht]
  \centering
  \includegraphics[width=0.49\textwidth]{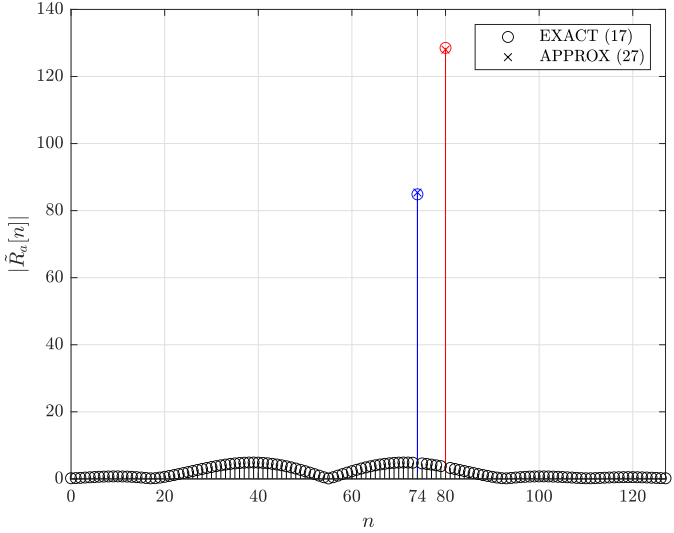}
  \caption{Illustration of demodulated LoRa symbol (non-coherent detection) for $SF=7$,   $a=80$, $a^-=42$ and two-path channel with
 $k_1=6$, $\alpha_0=1$ and $\alpha_1=0.7$.
  }
\label{fig:lora_magnitude_mpc}
\end{figure}

\begin{figure}[ht]
  \centering
  \includegraphics[width=0.49\textwidth]{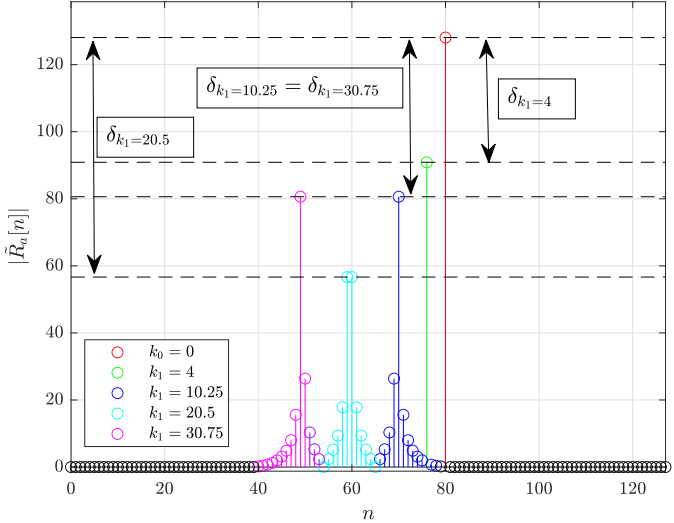}
  \caption{Illustration of demodulated LoRa symbol (non-coherent detection) with two-path non-aligned channel for $SF=7$, $a^- = a = 80$, $\alpha_0=1$, $\alpha_1=0.7$ and $k_1 = \{4,10.25,20.5,30.75\}$.
  }
\label{fig:lora_magnitude_mpc_non_aligned}
\end{figure}

\textcolor{black}{We point out that LoRa uses coded symbols in practice \textit{i.e.} using Gray and Hamming coding with Code Rate ($CR$) ranging from 4/5 to 4/8.
Only $CR = 4/7, 4/8$ can correct one bit per codeword.
Moreover, Gray coding implies that adjacent symbols differ only from one bit and thanks to the interleaving scheme used in practice \cite{burg}, LoRa off-by-one errors will be corrected statistically more often.
That is, applied to $MPC$ case, performance depends on the location of parasitic peaks \textit{i.e.} $k_i$ values.
When $k_i = 1$ (a situation that may be frequently encountered in practice) and with sufficiently high associated magnitude, the detected peak may be located at $n = a-1$ frequency index in low $SNR$ conditions, leading to an off-by-one error and thus improving the error correction capacity of the decoder.}

The correct symbol detection probability given $\Tilde{W}[a]$ can be expressed in term of  \textcolor{black}{$SER$} as:
\begin{equation}
        P_{d/W}^{(c)} = 1 - \mathbb{P}[\widehat{a} \ne a/H_a,H_{a^-},\Tilde{W}[a]]
    \label{eqPd1}
\end{equation}
where $c=1$ for self-$ISI$ case ($a^-=a$), and $c=2$ for the general $ISI$ case ($a^- \neq a$).
For the non-coherent detection scheme, a correct detection of $a$ must satisfy
the two following conditions on the $DFT$ magnitude (or squared $DFT$ magnitude):
\begin{itemize}
    \item The $K-1$ echoes located at $a-k_i$ ($i=1,\ldots,K-1$) must have a
      lower peak magnitude than the $a$-peak magnitude,
    \item Noise samples located at $n \neq \{a,\,a-k_i\}$ must have a lower peak
      magnitude than the $a$-peak magnitude.
\end{itemize}

This leads to:
\begin{equation}
  \begin{split}
  P_{d/W}^{(c)}&=\PP\left[|\Tilde{R}_a[a-k_i]|^2<|\Tilde{R}_a[a]|^2, \hspace{1em} i\in\{1,...,K-1\}\right.\\[-.2em]
  &\hspace{-2.4em}\text{and }\hfill\left.|\Tilde{R}_a[n]|^2<|\Tilde{R}_a[a]|^2,\hspace{1em} n\neq\{a,a-k_1,\ldots,a-k_{K-1}\}\right]
  \end{split}
\end{equation}

Independent events lead to:
\begin{equation}
  \begin{split}
    P_{d/W}^{(c)} &= \prod_{i=1}^{K-1} \mathbb{P}\left[ \left| d_i^{(c)} + \Tilde{W}[a-k_i]\right|^2 < \left|M\alpha_0 + \Tilde{W}[a]\right|^2\right] \\
        & \quad \prod_{\substack{n=0 \\ n \ne a \\ n \ne a-k_i
}}^{M-1} \mathbb{P}\left[ \left|\Tilde{W}[n]\right|^2 < \left|M\alpha_0 + \Tilde{W}[a]\right|^2 \right] 
 \end{split}
\end{equation}
where:
\begin{equation}
    d_i^{(c)} = 
    \begin{cases}
    M\Tilde{\alpha}_i(a)  & c = 1 \qquad\text{(self-$ISI$)}\\
    (M-k_i) \Tilde{\alpha}_i(a)  & c = 2 \qquad\text{($ISI$)}.
    \end{cases}
\label{eqdc}
\end{equation}

As noise samples are zero-mean circular Gaussian random variables, $|\Tilde{W}[n]|^2$ follows a centered chi-square $\chi^2$ distribution  with 2 degrees of freedom, and $|d_i^{(c)} + \Tilde{W}[a-k_i]|^2$ follows a non-centered chi-square $\chi^2_{NC}$ distribution with 2 degrees of freedom.
The non-centrality parameter $\lambda_i^{(c)}$ is:
\begin{equation}
    \begin{split}
        \lambda_i^{(c)} =
        \begin{cases}
        \frac{2M |\alpha_i|^2}{\sigma^2} & c=1 \qquad\text{(self-$ISI$)}\\
        \frac{2(M-k_i)^2 |\alpha_i|^2}{M \sigma^2} & c=2\qquad\text{($ISI$).}
        \end{cases}
    \end{split}
    \label{eqLambdaic}
\end{equation}

The probability of detection given $\Tilde{W}[a]$ for $ISI$ and self-$ISI$ cases can be expressed in terms of $CDF$s, this leads to:
\begin{equation}
  \begin{split}
P_{d/W}^{(c)} &=\prod_{i=1}^{K-1} F_{\chi^2_{NC}} \left( \frac{|M\alpha_0 + \Tilde{W}[a]|^2}{M \sigma_{\Re}^2} ; \lambda_i^{(c)} \right) \\
        & \qquad F_{\chi^2} \left( \frac{|M\alpha_0 + \Tilde{W}[a]|^2}{M \sigma_{\Re}^2} \right)^{M-K}.
\end{split} \label{eqPd1DetailAbs2}
\end{equation}

\textcolor{black}{Notice that \eqref{eqPd1DetailAbs2} and \eqref{eqLambdaic} do not depend on $a^-$ (thanks to the approximations (\ref{eq:Ya1_ISI_Approx}), (\ref{eq:Ya2_ISI_Approx}) and (\ref{eq:Ya3_ISI_Approx})), but they do not depend on $a$ either, then $\mathbb{P}[\widehat{a} \ne a/H_a,H_{a^-=a}]$ and $\mathbb{P}[\widehat{a} \ne a/H_a,H_{a^-\neq a}]$  do not depend on $a$ and $a^-$  (see (\ref{eqPd1})), which explains the simplification from \eqref{eqPe1}-\eqref{eqPe} to \eqref{eqPeApprox} for computing~$P_e$.
}

If we consider the $DFT$ magnitude instead of the squared $DFT$ magnitude, $P_{d/W}^{(c)}$ becomes:
\begin{equation}
    \begin{split}
 P_{d/W}^{(c)}  &  = \prod_{i=1}^{K-1} F_{Ri} \left(|M\alpha_0 + \Tilde{W}[a]| ; v_i^{(c)} ; \sigma_{\Re}\sqrt{M} \right) \\
        & \qquad F_{Ri} \left(|M\alpha_0 + \Tilde{W}[a]| ; 0 ; \sigma_{\Re}\sqrt{M}  \right)^{M-K}
    \end{split}
    \label{eqPd1DetailAbs}
\end{equation}
with $F_{Ri}(.,v,\sigma)$ the Ricean $CDF$ with non-centrality parameter $v$ and deviation $\sigma$.
$|\Tilde{W}[n]|$ follows a centered Ricean distribution (or Rayleigh distribution) and $|d_i^{(c)} + \Tilde{W}[a-k_i]|$ follows a non-centered Ricean distribution with non-centrality parameter $v_i^{(c)}$:
\begin{equation}
    \begin{split}
       v_i^{(c)} = 
       \begin{cases}
        M|\alpha_i| & c = 1\qquad\text{(self-$ISI$)}\\
        (M-k_i)|\alpha_i| & c = 2\qquad\text{($ISI$).}
    \end{cases}
    \end{split}
    \label{eqvji}
\end{equation}

Note that the results in \eqref{eqPd1DetailAbs2} and \eqref{eqPd1DetailAbs} are strictly equivalent.

Finally, if we consider the coherent detection scheme, $P_{d/W}^{(c)}$ becomes:
\begin{equation}
    \begin{split}
  P_{d/W}^{(c)} =& \prod_{i=1}^{K-1} F_{\mathcal{N}} \left( \Re \{ M\alpha_0 + \Tilde{W}[a]\}; \Re \{d_i^{(c)}\}; \sigma_{\Re}\sqrt{M}  \right) \\
        & \qquad F_{\mathcal{N}} \left( \Re \{ M\alpha_0 + \Tilde{W}[a]\}; 0 ;\sigma_{\Re} \sqrt{M}  \right)^{M-K}
    \end{split}
    \label{eqPd1DetailNormal}
\end{equation}

with $F_{\mathcal{N}}(.,\mu,\sigma)$ the normal $CDF$ with $\mu$-mean and $\sigma$-standard deviation parameters.
One may see that $\Re\{d_i^{(c)}\}$ depends on the transmitted symbol $a$.
We must use \eqref{eqPeApprox2} in order to compute $P_e$.
Each possible symbol $a$ must be taken into account, but this increases the computational complexity by a factor of $M$ in \eqref{eqPeApprox2} in comparison to \eqref{eqPeApprox} used for the non-coherent case.

\subsection{Expectation evaluation  via Gauss-Hermite integration}

Equations \eqref{eqPd1DetailAbs2}, \eqref{eqPd1DetailAbs} and \eqref{eqPd1DetailNormal} are used here for only one noise realisation $\Tilde{W}[a]$.
To obtain $P_e^{(c)}$, a mathematical expectation over $\Tilde{W}[a]$ should be performed:
\begin{equation}
        P_e^{(c)} = E[g^{(c)}(w)] 
    \label{eqPei}
\end{equation}
where $g^{(c)}(w)=1 - P_{d/W}^{(c)}(w)$ with $P^{(c)}_{d/W}(w)$ given by \eqref{eqPd1DetailAbs2} or \eqref{eqPd1DetailAbs} for the non-coherent case (and \eqref{eqPd1DetailNormal} for the coherent case) by replacing $\Tilde{W}[a]$ by $w$.
Note that to simplify notation, the hat symbol over $P_e$ is omitted in \eqref{eqPei} for $c=2$ (see \eqref{eqPeApprox} or \eqref{eqPeApprox2}).
In \cite{ferre1}, the authors propose to estimate \eqref{eqPei} by using a Monte-Carlo approach.
We propose instead to use the Gauss-Hermite procedure \cite{numerical_recipes} to efficiently compute numerically integral \eqref{eqPei}.
We obtain:
\begin{equation}
    \begin{split}
    P_e^{(c)} & = \frac{1}{\pi \sigma_w^2} \int_{\mathbb{C}} g^{(c)}(w) e^{- \frac{|w|^2}{\sigma_w^2}} dw \\
        & = \frac{1}{\pi} \int_{\mathbb{C}} g^{(c)}(\sigma_w w) e^{-|w|^2} dw \\
        & \approx \frac{1}{\pi} \sum_{n,m=1}^{N} g^{(c)}\left( \sigma \sqrt{M} (w_n + jw_m) \right) p_n p_m
    \end{split}
    \label{eqPeiDetail}
\end{equation}
where $w_i$ and $p_i$ for $i=1,\ldots N$ are respectively the nodes (abscissa) and weights of the $N$-points Gauss-Hermite quadrature rules.
To properly compute \eqref{eqPeiDetail}, $N$ must be sufficiently large (\textit{e.g.} $N=15$).

\section{LoRa user interference}\label{sec:lora_interf}

In this section, we derive based on previous developments a closed-form expression of LoRa  \textcolor{black}{$SER$} in the case of two LoRa users colliding in $AWGN$ channel.

\begin{figure}[ht]
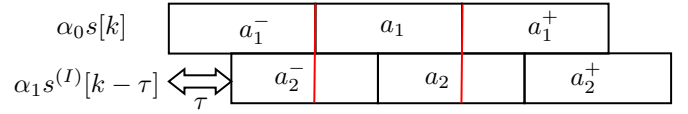

    \centering
 \input figures/demes4.tex
    \caption{LoRa user interference illustration.
    $ISI$ with symbols $a_2^-$ and $a_2$ from the interfer signal $s^{(I)}[k]$ appears in the detection interval of the desired current symbol~$a_1$.
}
    \label{fig:LoRaInterf}
\end{figure}

\begin{figure}[!t]
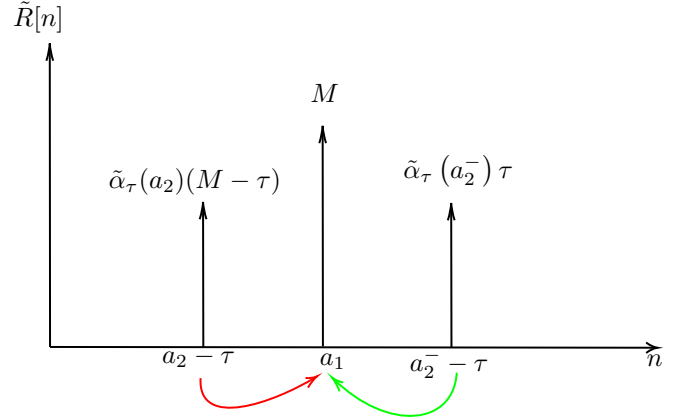

    \centering
    \input figures/demes5.tex
    \caption{LoRa $DFT$ for $a_2^- \ne a_2$.
    Colored arrows indicate  \textcolor{black}{additive interference} $A_2$ and $A_3$ cases.}
    \label{fig:LoRaInterfFFTNeq}
\end{figure}

\begin{figure}[!t]
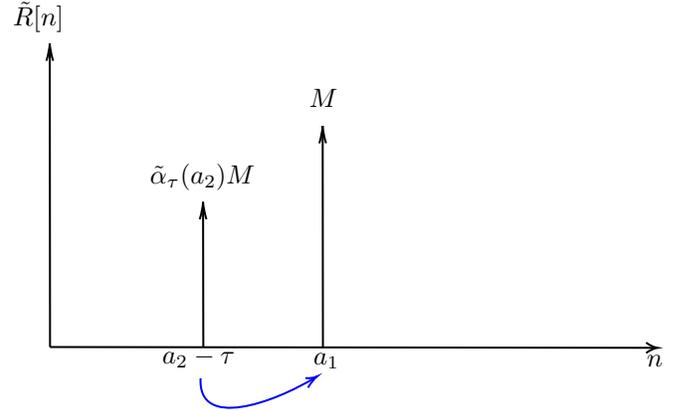

    \centering
    \input figures/demes6.tex
    \caption{LoRa $DFT$ for $a_2^- = a_2$.
    Arrow indicates the  \textcolor{black}{additive interference} $B_2$ case.}
    \label{fig:LoRaInterfFFTeq}
\end{figure}

\begin{figure*}[!b]
\normalsize
\setcounter{mytempeqncnt}{\value{equation}}
\hrulefill
\begin{equation}\label{eqRk2_IES}
    \begin{split}
    \Tilde{R}[n] & = M \delta[n - a_1] +\left\{ M_{\tau}[n;a_2^-] - M_{\tau}[n;a_2] \right\} (1-\delta[n - a_2 + \tau]) (1-\delta[n - a_2^- + \tau])\\
     &+ \left\{ (M - \tau) \Tilde{\alpha}_{\tau}(a_2) + M_{\tau}[a_2-\tau;a_2^-] \right\} \delta[n - a_2 + \tau]
     + \left\{ \tau \Tilde{\alpha}_{\tau}(a_2^-) - M_{\tau}[a_2^--\tau;a_2] \right\}\delta[n - a_2^- + \tau]
     \end{split}
\end{equation}
\setcounter{equation}{\value{mytempeqncnt}}
\end{figure*}

\subsection{Interference impact on $DFT$}

The model is slightly different from the $MPC$ model and is presented in Figure~\ref{fig:LoRaInterf}.
On the contrary of $MPC$, the delay of the interferer is not constrained to small values and is equally spread over the symbol duration \textit{i.e.} $\tau \in \{0,1,\ldots,M-1\}$.
By analogy with $MPC$ model, the user interference model corresponds to: $K=2$, with path delays $k_0=0$ and $k_1=\tau$, and path gains $\alpha_0=1$ and $\alpha_1\bydef{\alpha_{\tau}}=\sqrt{P_I}e^{j\phi}$.
Without loss of generality, we have considered normalized path gains over $\alpha_0$, then $\phi$ corresponds to the phase difference between second and first paths \textcolor{black}{and is considered uniformly distributed over $[0,2\pi[$}.
The interferer signal power is set with respect to Signal-to-Interference Ratio ($SIR$): $SIR = 1/P_I$.
The only difference between the two models appears in the $ISI$ term: $[a_1^-,a_1]$ (with $a_1=a$ the current symbol detection) for $MPC$ model while it corresponds to $[a_2^-,a_2]$ (with $a_2=a_1$ or $a_2\neq a_1$) for the interference model.
This difference with the interference symbol (\emph{i.e.} $a_2\neq a_1$) will exhibit more complexity in the  \textcolor{black}{$SER$} derivation for the LoRa user interference.
\textcolor{black}{The non-aligned interference case \textit{i.e.} $\tau$ a real value in $[0,M-1]$ will have similar effect as non-aligned $MPC$ (see Figure \ref{fig:lora_magnitude_mpc_non_aligned}).}

With similar manipulations as seen previously the received LoRa signal after dechirp process and $DFT$ is expressed as \eqref{eqRk2_IES} with:
\stepcounter{equation}
\begin{equation}
        M_{\tau}[n;a] = \Tilde{\alpha}_{\tau}(a) \sum_{k=0}^{\tau - 1} e^{2j\pi k/M(a - \tau - n)}
    \label{eqMtau_prec}
\end{equation}
 and:
 \begin{equation}
 \label{eq:alphat interf}
 \Tilde{\alpha}_{\tau}(a)=\sqrt{P_I}e^{j\phi}e^{-2j\pi
   \tau \frac{{a}}{M}} x_0[M - \tau].
\end{equation}

Note that in comparison with \eqref{eqRak2_IES}, \eqref{eqRk2_IES} exhibits the new peak amplitude $\tau \Tilde{\alpha}_{\tau}(a_2^-)$ at $n=a_2^- -\tau$ because for $\tau$ large enough the latter quantity isn't vanishing in comparison to $M$ at the correct $a_1$-peak.
For $a_2^-=a_2$, (\ref{eqRk2_IES}) simplifies to:
\begin{equation}
\label{eq:RnInterf2}
\Tilde{R}[n] = M \delta[n - a_1] + M  \Tilde{\alpha}_{\tau}(a_2) \delta[n - a_2 + \tau].
\end{equation}

Considering the special case $\tau=0$ to \eqref{eq:RnInterf2} yields:
\begin{equation}
    \Tilde{R}[n] = M \delta[n - a_1] + M  \alpha_{\tau} \delta[n - a_2].
\end{equation}

Depending on $a_1$, $a_2^-$ and $a_2$ values, different situations are possible as depicted in Figures \ref{fig:LoRaInterfFFTNeq} and \ref{fig:LoRaInterfFFTeq}, leading to the following five cases:

\begin{itemize}
\item $A_i$ cases in Figure \ref{fig:LoRaInterfFFTNeq} $(a_2^- \neq a_2)$:
\begin{itemize}
    \item $A_1 \triangleq \{a_2^- \neq a_2$, $a_2-\tau \neq a_1$ and $a_2^- - \tau \neq a_1\}$
    \item $A_2 \triangleq \{a_2^- \neq a_2$, $a_2-\tau = a_1$ and $a_2^- - \tau \neq a_1\}$\\ $\rightarrow$  \textcolor{black}{additive interference}
    \item $A_3 \triangleq \{a_2^- \neq a_2$, $a_2-\tau \neq a_1$ and $a_2^- - \tau = a_1\}$\\ $\rightarrow$  \textcolor{black}{additive interference}
    \end{itemize}
\item $B_i$ cases  in Figure~\ref{fig:LoRaInterfFFTeq}  $(a_2^- = a_2)$:
\begin{itemize} 
    \item $B_1 \triangleq \{a_2^- = a_2$ and $a_2-\tau \neq a_1\}$ 
    \item $B_2 \triangleq \{a_2^- = a_2$ and $a_2-\tau = a_1$\}\\ $\rightarrow$  \textcolor{black}{additive interference}
\end{itemize}
\end{itemize}

\textcolor{black}{The additive interference illustrated in Figure~\ref{fig:LoRaInterfFFTNeq}  for $(a_2^- \neq a_2)$  and Figure~\ref{fig:LoRaInterfFFTeq}  $(a_2^- = a_2)$ may be constructive or destructive depending on $\phi$ in \eqref{eq:alphat interf} and performance will be then improved or reduced.}

By considering $M_{\tau}[n^{\prime};a^{\prime}]$ terms in \eqref{eqRk2_IES} vanishing, the approximated  \textcolor{black}{$SER$} can be derived as a similar way as for the $MPC$ model.
However  \textcolor{black}{additive interference} cases are new and have to be considered.

\subsection{Approximated expressions of  \textcolor{black}{$SER$} under LoRa interference for non-coherent detection scheme}

It is straightforward to see that the approximated  \textcolor{black}{$SER$}, noted $ P_e^{(I)}$, under LoRa interference depends on $a_1$, $a_2^-$, $a_2$ symbol values and especially if  \textcolor{black}{additive interference} is present.
Similarly to \eqref{eqPeBayes}, $P_e^{(I)}$ is expressed as:
\begin{equation}
        P_e^{(I)} = \frac{1}{M^3} \sum_{a_1,a_2^-,a_2=0}^{M-1} \mathbb{P}[\widehat{a}_1 \ne a_1/H_{a_1},H_{a_2^-},H_{a_2}].
    \label{eqPeBayesInterf}
\end{equation}

Decomposing \eqref{eqPeBayesInterf} to the different cases aforementioned gives:
\begin{equation}
        P_e^{(I)} = P_e^{(A)} + P_e^{(B)}
    \label{eqPeBayesInterfDecomp}
\end{equation}
with:
\begin{equation}
    \begin{split}
        M^3 P_e^{(A)} = {} & M(M-1)(M-2) P_e^{(A_1)} \\
        &+ (M-1) \sum_{a_1=0}^{M-1} P_e^{(A_2)}(a_1) \\
        &+ (M-1) \sum_{a_1=0}^{M-1} P_e^{(A_3)}(a_1) 
    \end{split}
    \label{eqPeAInterf}
\end{equation}
and:
\begin{equation}
    \begin{split}
        M^3 P_e^{(B)} =& M(M-1) P_e^{(B_1)} + \sum_{a_1=0}^{M-1} P_e^{(B_2)}(a_1).
    \end{split}
    \label{eqPeBInterf}
\end{equation}

The assumption  $M_{\tau}[n^{\prime};a^{\prime}]\approx0$ allows us to avoid the nested summations over $a_2$ and $a_2^-$, which reduces drastically the complexity of $P_e^{(I)}$.
However the interference model brings more complexity in comparison with $MPC$ model because the peak-value at $a_1$ is \textit{reinforced} for $a_1=a_2-\tau$ or $a_1=a_2^--\tau$ and a summation over $a_1$ is required in $A_2$, $A_3$ and $B_2$ cases as seen in \eqref{eqPeAInterf} and \eqref{eqPeBInterf}.
For example, for the $B_2$ case, the peak-value at $a_1$ is equal to $M$ for $a_1\neq a_2-\tau$ while for $a_1=a_2-\tau$ it is equal to $M\tilde{\alpha}_{\tau}(a_1+\tau)+M$, which depends on $a_1$.

For $\tau=0$, only the $B_i$ cases have to be considered but the peak-value at $a_1$ of the \textcolor{black}{additive interference} $B_2$ case is $M(1+\alpha_{\tau})$, which does not depend on $a_1$. Therefore $P_e^{(B_2)}(a_1)=P_e^{(B_2)}$ and the  \textcolor{black}{$SER$} for $\tau=0$ is:
\begin{equation}
     M P_e^{(I)} = (M-1) P_e^{(B_1)} + P_e^{(B_2)}.
\end{equation}

By following the same development used for $MPC$, the $g^{(c)}$-functions in \eqref{eqPeiDetail} used to derive  \textcolor{black}{$SER$} are obtained via:\\[-1.5em]
\begin{itemize}
    \item $g^{(C_i)}(\sigma\sqrt{M}w)=1-P_{d/W}^{(C_i)}(\sigma\sqrt{M}w)$\\ for $C_i=\{A_1, B_2\}$,\\[-.5em]
    \item $g^{(C_i)}(a_1,\sigma\sqrt{M}w)=1-P_{d/W}^{(C_i)}(a_1,\sigma\sqrt{M}w)$\\ for  $C_i=\{A_2, A_3, B_2\}$,
\end{itemize}
 with $P_{d/W}^{(C_i)}(\sigma\sqrt{M}w)$ or $P_{d/W}^{(C_i)}(a_1,\sigma\sqrt{M}w)$ equals to $P_{d/W}^{(C_i)}$ given in \eqref{eqPdA1}-\eqref{eqPdB2}.
 Note that for $C_i=\{A_2, A_3, B_2\}$ the $g^{(c)}$-functions depend on $a_1$ for the evaluation of $P_e^{(C_i)}(a_1)$.

The detection probabilities $P_{d/W}^{(C_i)}$ used in the  $g^{(c)}$-functions for non coherent detection scheme are summarized below:
\setlength\arraycolsep{0.1pt}
\begin{IEEEeqnarray}{rCl}
P_{d/W}^{(A_1)} &=& F_{\chi_{NC}^2} \left( d^{} ; \lambda_{\tau}^{(1)} \right) F_{\chi_{NC}^2} \left( d^{} ; \lambda_{\tau}^{(2)} \right)  F_{\chi^2} \left( d^{} \right)^{M-3}\hspace{1.5em}\IEEEyesnumber\label{eqPdA1}\\
P_{d/W}^{(A_2)} &=& F_{\chi_{NC}^2} \left( d_{\tau}^{(A_2)} ; \lambda_{\tau}^{(2)} \right) F_{\chi^2} \left( d_{\tau}^{(A_2)} \right)^{M-2}\IEEEyessubnumber \label{eqPdA2} \\
P_{d/W}^{(A_3)} &=& F_{\chi_{NC}^2} \left( d_{\tau}^{(A_3)} ; \lambda_{\tau}^{(1)} \right) F_{\chi^2} \left( d_{\tau}^{(A_3)} \right)^{M-2} \IEEEyessubnumber\label{eqPdA3} \\
P_{d/W}^{(B_1)} &=& F_{\chi_{NC}^2} \left( d ; \lambda_{}^{(3)} \right) F_{\chi^2} \left( d \right)^{M-2} \IEEEyessubnumber\label{eqPdB1} \\
P_{d/W}^{(B_2)} &=& F_{\chi^2} \left( d_{\tau}^{(B_2)}\right)^{M-1}
\IEEEyessubnumber\label{eqPdB2}
\end{IEEEeqnarray}
with:\vspace{-1em}
\begin{IEEEeqnarray}{rCl}
    d &=& \frac{2|\sqrt{M} + w\sigma|^2}{\sigma^2}\IEEEyesnumber \label{eqd} \\
    d_{\tau}^{(A_2)} &=& \frac{2|\sqrt{M} + \sqrt{\frac{(M-\tau)^2}{M}}\Tilde{\alpha}_{\tau}(a_1+\tau) + w\sigma|^2}{\sigma^2}\IEEEyessubnumber \label{eqdrenfg} \\
    d_{\tau}^{(A_3)} &=& \frac{2|\sqrt{M} + \sqrt{\frac{\tau^2}{M}}\Tilde{\alpha}_{\tau}(a_1+\tau) + w\sigma|^2}{\sigma^2}\IEEEyessubnumber \label{eqdrenfd} \\
    d_{\tau}^{(B_2)} &=& \frac{2|\sqrt{M} + \sqrt{M}\Tilde{\alpha}_{\tau}(a_1+\tau) + w\sigma|^2}{\sigma^2}\IEEEyessubnumber \label{eqdrenfgeq}
\end{IEEEeqnarray}
and:\vspace{-1em}
\begin{IEEEeqnarray}{rCl}
    \lambda_{\tau}^{(1)} &=& \frac{2(M-\tau)^2 P_I}{M \sigma^2}\IEEEyesnumber \label{eqlambdatau1} \\
    \lambda_{\tau}^{(2)} &=& \frac{2\tau^2 P_I}{M \sigma^2}\IEEEyessubnumber \label{eqlambdatau2} \\
    \lambda_{}^{(3)} &=& \frac{2M P_I}{\sigma^2}.\IEEEyessubnumber \label{eqlambdatau3}
\end{IEEEeqnarray}

Note that for the special case $\tau=0$, $d_{\tau}^{(B_2)}=d_{0}^{(B_2)}=2|\sqrt{M}(1+{\alpha}_{\tau}) + w\sigma|^2/\sigma^2$.

Computational complexity of $P_e^{(I)}$ for $\tau\neq0$ comes from the summations over all the possible symbols $a_1$ in \eqref{eqPeAInterf} and \eqref{eqPeBInterf}. 
However for even $\tau$ values, these summations can be significantly reduced as shown in  proposition \ref{prop:complexity}.

\begin{prop}
\label{prop:complexity}
Complexity of performance evaluation in (\ref{eqPeBayesInterfDecomp}) can be reduced by a factor of $2^n$ in case of delay $\tau$ in the following form: $\tau=k2^n$ with $k$ an odd number and $n$ an integer greater than~0 ($n\in\{1,\ldots,SF-1\}$).
\end{prop}

\begin{proof}
To reduce complexity we can consider only the distinct values with respect to LoRa symbol $a_1$ to perform the summations $P_e^{(C_i)}(a)$ with $C_i=\{A_2,A_3,B_2\}$ in \eqref{eqPeAInterf} and \eqref{eqPeBInterf}.
$P_e^{(C_i)}(a)$ depends directly on $P_{d/M}^{(C_i)}$ where LoRa symbol $a_1$ appears in $d_{\tau}^{(C_i)}$ in \eqref{eqdrenfg}, \eqref{eqdrenfd} and \eqref{eqdrenfgeq}.
By replacing $\tau=k2^n$ into the complex exponential term of $\Tilde{\alpha}_{\tau}(a_1+\tau)$ we obtain $\Tilde{\alpha}_{\tau}(a_2) = \sqrt{P_I}e^{j\phi}e^{-2j\pi k a_2/M_1} x_0[M-\tau]$ with $M_1=M/2^n$ and $a_2=a_1+\tau$.
For $k$ odd, we obtain $M_1$ distinct values when $a_2$ is varying from 0 to $M_1-1$.\footnote{Note that for $\Tilde{\alpha}_{\tau}(a_1+\tau)$, the set of symbols is shifted by $\tau$ but it leads to the same set of values.
It just produces a circular shift.}
Doing all the set of symbols $\{0,1,,M-1\}$ leads to $2^n$ times the same values obtained with $\{0,1,,M_1-1\}$.
Hence we can replace $\sum_{a_1=0}^{M-1} P_e^{(C_i)}(a_1)$ by $2^n\sum_{a_1=0}^{M_1-1} P_e^{(C_i)}(a_1)$ with $M_1=M/2^n=2^{SF-n}$.
The complexity is then reduced by a factor of $2^n$.
For example, for $\tau=3\times 2^5=96$ and $SF=7$, the summations reduce to $2^{SF}/2^5=4$ terms (\textit{i.e.} $\sum_{a=0}^{M-1}P_e^{(C_i)}(a)=2^{SF-2}\sum_{a=0}^3 P_e^{(C_i)}(a)$ for $C_i=\{A_2,A_3,B_2\}$).
\end{proof}

From the theoretical approximated  \textcolor{black}{$SER$}, we can derive interesting results about $\tau$- and $\phi$-influence on performance.

\subsubsection{Influence of $\tau$}
 \textcolor{black}{$SER$} performance in terms of $\tau$ is symmetric about the axis $\tau=M/2$ as shown in proposition~\ref{prop:sym}.

\begin{prop}
\label{prop:sym}
The approximated  \textcolor{black}{$SER$} expression $P_e^{(I)}$ in \eqref{eqPeBayesInterfDecomp} is invariant by changing $\tau$ by $M-\tau$.
\end{prop}

\begin{proof}
$P_e^{(B_1)}$ doesn't depend on $\tau$ (see \eqref{eqPdB1}).
It is straightforward to see from \eqref{eqPdA1} that $P_e^{(A_1)}$ does not change if considering $\tau$ or $M-\tau$ delays.
$P_e^{(B_2)}$ depends on $P_{d/W}^{(B_2)}$ where only the term $\Tilde{\alpha}_{\tau}(a_1+\tau)$ in (\ref{eqdrenfgeq}) depends on $\tau$.
One can verify $\Tilde{\alpha}_{\tau}(a_2) = \sqrt{P_I}e^{j\phi}e^{-2j\pi\tau a_2/M} x_0[M-\tau]$ and $\tilde{\alpha}_{M-\tau}(a_2) = \sqrt{P_I}e^{j\phi}e^{2j\pi\tau a_2/M} x_0[M-\tau]$ with $a_2=a_1+\tau$.
The only difference is the clockwise versus anticlockwise rotation in the complex exponential, but for $a_2$ in the set $\{0,1,\ldots,M-1\}$, the direction of rotation doesn't change the set of values of $\Tilde{\alpha}_{\tau}(a_2)$ or $\Tilde{\alpha}_{M-\tau}(a_2)$.
By shifting the set by $\tau$ (\textit{i.e.} $\{\tau,\tau+1,\ldots,\tau+M-1\}$) it gives the same set of values (it just produces a circular shift).
As the evaluation of $P_e^{(B)}$ performs a sum over all symbols $a_1$ in \eqref{eqPeBInterf}, the result of $\sum_{a=0}^{M-1}P_e^{(B_2)}(a)$ for a given delay path $\tau$ or $M-\tau$ leads to the same result.
It remains to examine the  $A_2$ and $A_3$ cases. 
From \eqref{eqPdA2} and \eqref{eqPdA3}, by changing $\tau$ by $M-\tau$ leads to permute $P_{d/W}^{(A_2)}$ to $P_{d/W}^{(A_3)}$ except for the clockwise versus anticlockwise rotation of $\Tilde{\alpha}_{\tau}(a_1+\tau)$ in \eqref{eqdrenfg} and \eqref{eqdrenfd}.
Hence, by changing $\tau$ by $M-\tau$ in the evaluation of $P_e^{(A)}$, the sum  $\sum_{a=0}^{M-1}P_e^{(A_2)}(a)$ turns into $\sum_{a=0}^{M-1}P_e^{(A_3)}(a)$, and vice versa.
The results for $M-\tau$ are then equivalent to those obtained with $\tau$.
\end{proof}

\subsubsection{Influence of $\phi$} \label{subsubsec:phi}
Contrary to performance of the non-coherent detection for the $MPC$ model, performance for the interference model depends on $\phi$ because the additive interference terms in \eqref{eqdrenfg}, \eqref{eqdrenfd} and \eqref{eqdrenfgeq} depend on $\phi$ via $\Tilde{\alpha}_{\tau}(a)$ given in \eqref{eq:alphat interf}.
However, its influence for LoRa modulation with $SF$ in $\{7,8,\ldots\}$  has a very low impact on performance.
By examining the complex exponential term of $\Tilde{\alpha}_{\tau}(a_1+\tau)$ with respect to $a_1$, the possible angles $\psi_{\tau}(m)$ of  $\Tilde{\alpha}_{\tau}(a_1+\tau)$ are:
\begin{equation}
\label{eq:psi}
 \psi_{\tau}(m)=2\pi\frac{m}{M_1}+\theta_{\tau}+\phi
\end{equation}

for $m=0,\ldots,M_1-1$ where $\theta_{\tau}$ is the angle of $x_0[M-\tau]$ and $M_1=M/2^n$ for $\tau$ expressed as $\tau=k2^n$ with $k$ odd (see proposition \ref{prop:complexity}). 
By replacing $\phi$ by $\phi+2\pi/M_1$ in (\ref{eq:psi}), we obtain the same equivalent set of $\Tilde{\alpha}_{\tau}(a_1+\tau)$ for $a_1=0,\dots,M-1$, which involves  \textcolor{black}{$SER$} performance is $\frac{2\pi}{M_1}$-periodic in term of $\phi$.
For $\tau$ an odd number, $M_1=M=2^{SF}$ the domain of study of $\phi$ becomes very small and its influence tends to be negligible
\textcolor{black}{for the typical LoRa $SF\ge7$
values}.

\textcolor{black}{
A fine analysis based on the variation of the minimum distance ($d_{\min}$) of (\ref{eqdrenfg})-(\ref{eqdrenfgeq}) with respect to $\phi$ allows us to obtain the  two extreme $\phi$ values: $\{\phi^{\min},\phi^{\max}\}$ where $\phi^{\min}$ leads to the worst performance (\textit{i.e.} $\min_{\phi}d_{\min}$), and $\phi^{\max}$ to the best performance (\textit{i.e.} $\max_{\phi}d_{\min}$). 
For the sake of brevity and clarity, we give only the results on the $\phi^{\min}$ and $\phi^{\max}$ values. For $\tau$ even: $\phi^{\min}=0$ and  $\phi^{\max}=\pi/M_1$, and for $\tau$ odd: $\phi^{\min}=\pi/M$ and  $\phi^{\max}=0$. Simulation part (Section VII-B) will confirm that $\phi$ has a very low influence on $SER$ except at $\tau=M/2$ where a significant difference between $\phi=\pi/M_1=\pi/2$ and $\phi=0$ can be observed. 
}

\subsubsection{Application of the study about $\phi$ on the $MPC$ model} \label{secsec:phiMPC}
We have shown in subsection \ref{subsec:Pd} that theoretical performance for the non-coherent detection is independent of the phase $\phi_i$ of the complex path gain $\alpha_i$ because  \textcolor{black}{$SER$} performance depends only on the modulus $|\alpha_i|$.
However for the coherent detection scheme, it depends on $\Re\{\Tilde{\alpha}_{\tau}(a)\}$ in \eqref{eqdc} where the phase $\phi_i$ is present.
By applying a similar analysis as \ref{subsubsec:phi} about the $\phi$ influence for the interference model, we obtain that the domain of study for the variation of $\phi_i$ is limited at $[0,\pi/M_1]$.
For the $MPC$, we supposed $\tau\ll M$ (contrary to the interference model) then $M_1$ is large enough for even $\tau$ and $M_1=M$ for odd $\tau$.
The domain of study \textcolor{black}{$[0,\pi/M_1]$} tends to 0 for the typical LoRa $SF\ge7$ values. 
We conclude that $\phi_i$ has a very low impact on coherent performance for the $MPC$ model, which is confirmed by simulations.

\section{simulation results}\label{sec:Simu}

We present in this section several simulation results to validate theoretical  \textcolor{black}{$SER$} performance set forth herein.
The signal-to-noise ratio ($SNR$) presented in the  \textcolor{black}{$SER$} plots is defined as $1/\sigma^2$.
Simulation results are done with both coherent and non-coherent detection schemes.
Two examples of $MPC$ are considered, and LoRa interference is investigated.

\textcolor{black}{Simulations are performed by using the discrete-time $MPC$ or the aligned interference model, and by applying the $DFT$ on the received dechirping signal, which is corresponding to the exact $DFT$ expressions derived in \eqref{eqRak2_IES} and \eqref{eqRk2_IES} for both scenarios, whereas for the theoretical assessment the $M_i[\textbf{.},\textbf{.}]$ terms in \eqref{eqRak2_IES} and \eqref{eqRk2_IES} are neglected.
}

\subsection{Performance evaluation under MPC}

\subsubsection{Numerical validation}\label{subsec:valid}
We first consider the two-path ($K=2$) channel $C_1(z) = 1 + \alpha_1 z^{-k_1}$. 
The $g^{(c)}$-function in (\ref{eqPeiDetail}) becomes:
\begin{equation}
\begin{split}
  g^{(c)}(\sigma\sqrt{M}{w})&=1-F_{\chi_{NC}^2}\left(\frac{2|\sqrt{M}+\sigma {w}|^2}{\sigma^2}\;;\;\lambda_1^{(c)}\right)\\
  &F_{\chi^2}\left(\frac{2|\sqrt{M}+\sigma {w}|^2}{\sigma^2}\right)^{M-2}
\end{split}
  \label{eq:chi2echo}
\end{equation}
with $\lambda_1^{(c)}$ given in (\ref{eqLambdaic}).  We have: $\lambda^{(1)}_1
= 2M |\alpha_1|^2/\sigma^2$ and $\lambda_1^{(2)} = 2(M - k_1)^2|\alpha_1|^2/(M
\sigma^2)$. 
The numerical evaluation of  \textcolor{black}{$SER$} is obtained by using the Gauss-Hermite procedure \eqref{eqPeiDetail} with \eqref{eq:chi2echo} for $c=1$ and $c=2$, and then by using \eqref{eqPeApprox}.

\textcolor{black}{Figure \ref{fig:simu5} compares simulation and theoretical $SER$ results for both non-coherent and coherent schemes.
We consider here $SF=7$, two different delays and path gains, $k_1 \in \{1,10\}$ and $\alpha_1 \in \{0.7,0.9\}$, respectively.
In the figure, markers and line styles indicate respectively delay and path gain values.
Green curves indicate simulations.
From the figure, we may see that simulation results fit very well with theoretical $SER$ expression for both coherent and non-coherent cases.
The very slight bias is due to simplifications done in theoretical $SER$ computation process.
This confirms good adequacy between theoretical and simulation results.}

\begin{figure}[ht]
  \centering
  \includegraphics[width=0.49\textwidth]{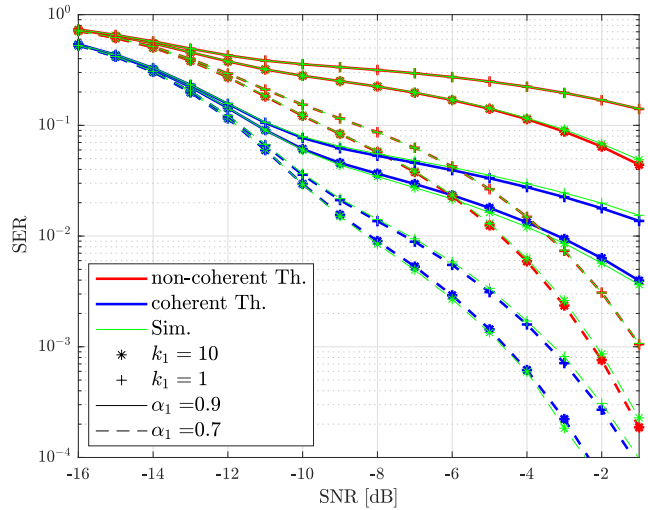}
  \caption{Theoretical \textcolor{black}{(Th.)} and simulation \textcolor{black}{(Sim.)} \textcolor{black}{$SER$} performance comparison over the two-path channel for non-coherent and coherent schemes.
  $SF=7$, $\alpha_1 = \{0.7,0.9\}$ and $k_1 = \{1,10\}$.
  }
\label{fig:simu5}
\end{figure}

It is worth noting in all rigor, performance depends on the phase $\phi_i$ of the path gain $\alpha_i$ as shown in \eqref{eqRak2_IES} where $\phi_i$ appears also in the simplified terms  $M_i[.;.]$.
However, by simplifying the $M_i[.;.]$ terms to derived the theoretical expressions, the non-coherent detection performance depend only of $|\alpha_i|$, and for the coherent case, we have shown in \ref{secsec:phiMPC} that the influence of $\phi_i$ is negligible.
We have considered the two extreme cases $\phi=0$ and $\phi=\pi/M_1$ (see \ref{subsubsec:phi}) for the angle of $\alpha_1$ in simulations, no difference was observed in the Symbol Error Rate ($SER$) (for the sake of clarity only $\phi=0$ is presented).\\

\subsubsection{Performance degradation for the two-path channel}\label{sec:degradation}

In this paragraph we quantify the performance degradation due to the presence of one echo at different delays and amplitudes in comparison with the ideal one-path channel.

Figure \ref{fig:simu} shows theoretical  \textcolor{black}{$SER$} results for $SF=7$ and channel $C_1(z)$ with $\alpha_1 \in \{0,\,0.2,\,0.4,\, 0.6,\, 0.8,\, 0.9\}$ and $k_1\in\{1,3,5,7,9,11\}$.
$\alpha_1=0$ corresponds to the one-path channel.

\begin{figure}[ht]
  \centering
  \includegraphics[width=0.49\textwidth]{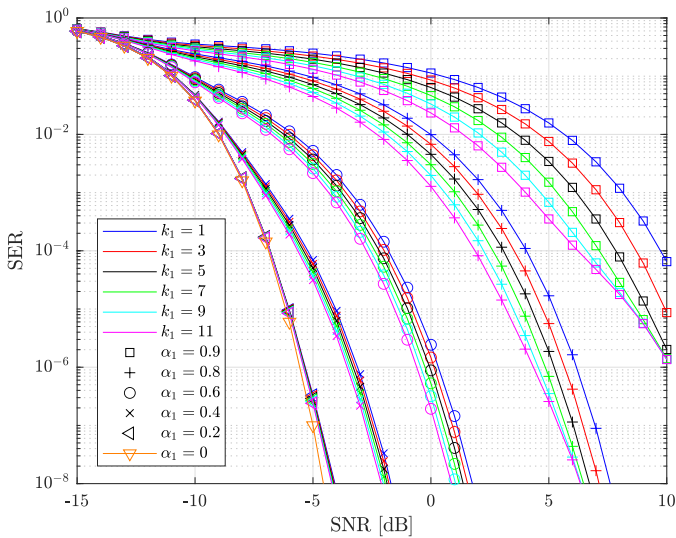}
  \caption{$\alpha_1$ and $k_1$ channel parameters influence on  \textcolor{black}{theoretical $SER$} performance with $SF=7$ over 
the two-path channel for $\alpha_1 = \{0,0.2,0.4,0.6,0.8,0.9\}$ and $k_1 = \{1,3,5,7,9,11\}$.}
\label{fig:simu}
\end{figure}

We note that performance loss is significant for $\alpha_1 \ge 0.4$.
This reveals that the more the echo is far from direct path (\textit{i.e.} $k_1$ bigger), the better performance are, confirming prediction from \eqref{eqRak2_IES}.
When $\alpha_1 \le 0.4$, $k_1$ has a negligible impact on  \textcolor{black}{$SER$}. 
We highlight that for $\alpha_1 = \{0.8,0.9\}$ and $k_1=11$, performance converges towards those for $k_1 = 9$.
These values are confirmed via $SER$ simulations (not shown in the figure for clarity). 
At high $SNR$ and $k_1$ large enough (e.g. $k_1=11$), an error when $a^-=a$ exhibits the interference peak value $M\Tilde{\alpha_1}(a)$ at frequency index $a-k_1$ (value required for $P_e^{(1)}$) whereas for $a^-\neq a$ the interference peak value at $a-k_1$ is smaller and equals to  $(M-k_1)\Tilde{\alpha_1}(a)$ (value required for $P_e^{(2)}$).
Even if the event $a^-=a$ is rarer, $P_e^{(1)}$ is much higher than $P_e^{(2)}$, and becomes a dominant term in \eqref{eqPeApprox} as shown the changing behavior in Figure \ref{fig:simu}.
\textcolor{black}{We may see that the bandwidth parameter $B$ has also impact on $SER$ performance.
For example, let us consider a two-path channel with an echo fixed at $8 \mu s$.
This represents a tap delay $k_1 = \{4,2,1\}$ for $B = \{500,250,125\}$ kHz, respectively.
According to Figure~\ref{fig:simu}, we have seen that for a fixed echo path gain, performance is better if its delay increases.
Therefore, we expect to have slightly better performance for higher bandwidths.
}

Figure \ref{fig:simu2} shows  \textcolor{black}{theoretical} \textcolor{black}{$SER$} for different $SF$ with $k_1=1$ at different path gains $\alpha_1$.
Performance increases with higher $SF$.
Indeed, for $\alpha_1 = 0$ (ideal one-path channel), performance gain is about 3.5~dB when increasing $SF$ by 1.
We observe the same gain (around 3.5~dB) for each amplitude $\alpha_1 \neq0$ when increasing $SF$ by 1.
$\alpha_1$ has obviously a huge impact leading to a performance loss about 12~dB (measured at  \textcolor{black}{$SER$} $=10^{-8}$) from $\rho = 0$ to $\rho = 0.8$, whatever $SF$ is.
Table \ref{tab:pertes} gives the performance losses in dB (at  \textcolor{black}{$SER$} $=10^{-8}$) for a given $SF$ between two arbitrary values of $\alpha_1$.
We globally observe the same performance loss whatever $SF$ is.
\textcolor{black}{We conclude that $SF$ is a crucial parameter to make LoRa resilient to multi-path environments, that is coherent with conclusions drawn in \cite{staniec}.}

\begin{figure}[ht]
  \centering
  \includegraphics[width=0.49\textwidth]{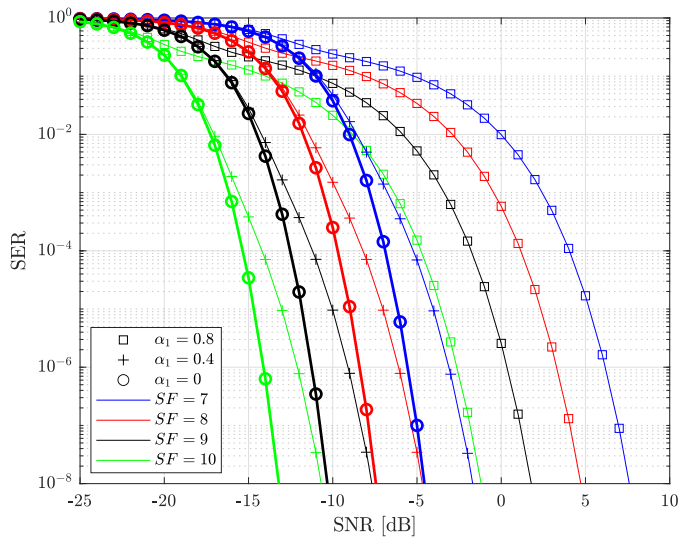}
  \caption{$\alpha_1$ and $SF$ parameters influence on  \textcolor{black}{theoretical $SER$} performance over
    the two-path channel for $SF = \{7,8,9,10\}$ with $k_1 =1$ and $\alpha_1 =
    \{0,\,0.4,\,0.8\}$.
    Legend: a different marker for each path gain $\alpha_1$, a different color for each $SF$, and thick lines for the one-path channel ($\alpha_1=0$).
}
\label{fig:simu2}
\end{figure}

\begin{figure}[ht]
  \centering
  \includegraphics[width=0.49\textwidth]{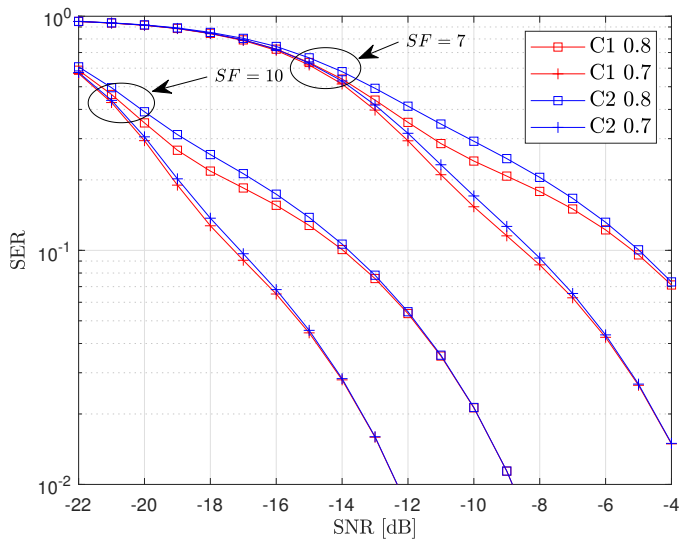}
  \caption{ \textcolor{black}{Theoretical $SER$} comparison between the two-path channel (C1) an the exponentially decreasing channel (C2) for $SF = \{7,10\}$ and $\rho = \{0.7,0.8\}$.}
\label{fig:simu3}
\end{figure}

\begin{table}[ht]
\begin{center}
\begin{tabular}{@{}l||l|l|l|l|l|l@{}}
$SF$ & 7 & 8 & 9 & 10 & \textcolor{black}{11} & \textcolor{black}{12} \\\hline 
$\Delta_1$ $\alpha_1=0\rightarrow0.4$ & 2.89 & 2.76 & 2.64 & 2.51 & \textcolor{black}{2.40} & \textcolor{black}{2.31} \\\hline
$\Delta_2$ $\alpha_1=0.4\rightarrow0.5$ & 1.58 & 1.57 & 1.58 & 1.58 & \textcolor{black}{1.60} & \textcolor{black}{1.59} \\\hline
$\Delta_3$ $\alpha_1=0.5\rightarrow0.6$ & 1.89 & 1.91 & 1.92 & 1.91 & \textcolor{black}{1.90} & \textcolor{black}{1.93} \\\hline
$\Delta_4$ $\alpha_1=0.6\rightarrow0.7$ & 2.42 & 2.46 & 2.47 & 2.48 & \textcolor{black}{2.49} & \textcolor{black}{2.47} \\\hline
$\Delta_5$ $\alpha_1=0.7\rightarrow0.8$ & 3.41 & 3.46 & 3.51 & 3.50 & \textcolor{black}{3.50} & \textcolor{black}{3.53} \\\hline
\begin{tabular}{l}cumulative loss\\ $\alpha_1=0\rightarrow0.8$\end{tabular} & 12.19 & 12.16 & 12.12 & 11.98 & \textcolor{black}{11.89} & \textcolor{black}{11.83} \\\hline  
\end{tabular}
\end{center}
\caption{Performance loss in dB $\Delta_i$ (measured values) at  \textcolor{black}{$SER$} equals to $10^{-8}$ between $\alpha_1=x\rightarrow y$ for a given $SF$.
To find the cumulative loss between 2 values, losses have to be added.}
\label{tab:pertes}
\end{table}

\subsubsection{Performance evaluation in the exponential decay channel}

We consider now the exponential decay channel $C_2(z)=\sum_{i=0}^{K-1} \rho^i z^{-i}$.
The maximum path number $K$ is determined by satisfying the condition $|\rho|^K \le 0.2$.
The $g^{(c)}$-function in \eqref{eqPeiDetail} used for  \textcolor{black}{$SER$} evaluation becomes:
\begin{equation}
\begin{split}
  g^{(c)}(\sigma\sqrt{M} {w})&=1-\prod_{i=1}^{K-1} F_{\chi_{NC}^2}\left(\frac{2|\sqrt{M}+\sigma {w}|^2}{\sigma^2}\;;\;\lambda_i^{(c)}\right)\\
& F_{\chi^2}\left(\frac{2|\sqrt{M}+\sigma {w}|^2}{\sigma^2}\right)^{M-K}
\end{split}
 \label{eq:chi2echo2}
\end{equation}
with $\lambda_i^{(1)}=2M|\alpha_i|^2/\sigma^2=2M|\rho|^{2i}/\sigma^2 $ and $\lambda_i^{(2)}=2(M-i)^2|\rho|^{2i}/(\sigma^2M)$ for $i=1,\ldots,K-1$. 

Performance over channel $C_2(z)$ are very close to the channel $C_1(z)$ one's presented in Figure \ref{fig:simu2}, which seems to say that only the first most significant path degrades  \textcolor{black}{$SER$} performance.
For $\rho\le0.6$ differences between $C_1(z)$ and $C_2(z)$ are negligible whatever $SF$.
However, for $\rho = \{0.7,0.8\}$, Figure \ref{fig:simu3} focuses on the  \textcolor{black}{$SER$} range where we could observe a difference between the two channels.
We observe a slight additional degradation with $C_2(z)$ whatever $SF$ (only $SF = \{7,10\}$
are considered in the figure).

\subsection{Performance evaluation under LoRa interference}

We evaluate LoRa performance of the non-coherent detection scheme in presence of LoRa interference (with same $SF$) in function of the interferer delay $\tau$ for $SIR_{dB} = 3$, $SF = \{8,\,10,\,12\}$ and different $SNR_{dB}$ values.

We consider multiple even $\tau$ values.
The $\tau$-range is given by $\tau = \{0,\,\tau_{step}\,,2\tau_{step}\,,3\tau_{step}\,,\ldots,(2^5-1)\tau_{step}=M-\tau_{step}\}$ with the
$\tau$-step value at $\tau_{step} = 2^{SF-5}$.
Notice that $16\tau_{step}=2^{SF-1}=M/2$.

The channel phase $\phi$ of the interferer user is chosen at $\phi=\phi^{\min}=0$ and $\phi = \phi^{max}=\pi/M_1$, that respectively corresponds  to the unfavorable and the favorable performance, for even $\tau$ (see paragraph \ref{subsubsec:phi}).
Note that $M_1=M/2^n$ with $\tau$ in the following form $\tau=k2^n$ (odd $k$).
Theoretical  \textcolor{black}{$SER$} is compared to simulated $SER$ in Figure \ref{fig:simu_LoRa_interf} for $\phi=\phi^{\max}$ and $\phi=\phi^{\min}$.
We observe a significant difference for theoretical \textcolor{black}{$SER$} (Th. $\phi^{\max}$ versus Th. $\phi^{\min}$) or simulated $SER$ (Sim. $\phi^{\max}$ vs. Sim. $\phi^{\min}$) only at $\tau=M/2$. 
These results confirm that $\phi$ has no influence on  \textcolor{black}{$SER$} except for the special case $\tau=M/2$ as it was discussed in paragraph \ref{subsubsec:phi}.
It is worth noting that progressively increasing $SF$ values reduces the bias between our theoretical and simulated results for $\tau$ around $M/2$ (Th. vs. Sim. for $\phi^{\min}$ or Th. vs. Sim. for $\phi^{\max}$).

Otherwise, from $\tau=\tau_{step}$ (or from its symmetric value, $\tau=M-\tau_{step}$) to $\tau=M/2$, we observe in Figure~\ref{fig:simu_LoRa_interf} a big variation in term of  \textcolor{black}{$SER$}.
The  \textcolor{black}{$SER$} variation is bigger as $SF$ increases.

\begin{figure}[ht]
  \centering
  \includegraphics[width=0.49\textwidth]{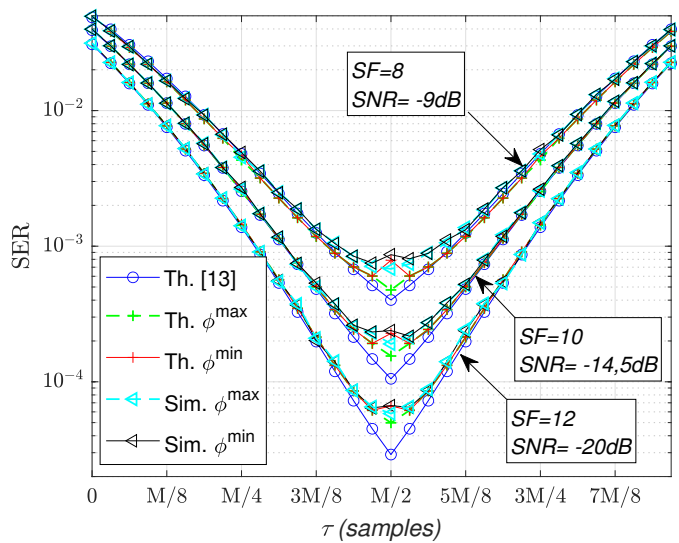}
  \caption{\textcolor{black}{Theoretical (Th.) and simulation (Sim.) performance} comparison of LoRa interference versus even $\tau$ values for $SIR_{dB}=3$.
  The extreme phase values $\phi^{\max}=\pi/M_1$ and $\phi^{\min}=0$ are considered for our theoretical (Th.) and simulated (Sim.) results.
  Theoretical  \textcolor{black}{$SER$} of \cite{joerg5} is also reported for fixed $\tau$ values.
  For comparison with LoRa modulation over $AWGN$,  \textcolor{black}{$SER$} is 0.9781 $\times 10^{-5}$, 0.4788 $\times 10^{-5}$ and 0.1792 $\times 10^{-5}$ for $SF=8$, 10 and 12, respectively.  
  }
\label{fig:simu_LoRa_interf}
\end{figure}

\begin{figure}[ht]
  \centering
  \includegraphics[width=0.49\textwidth]{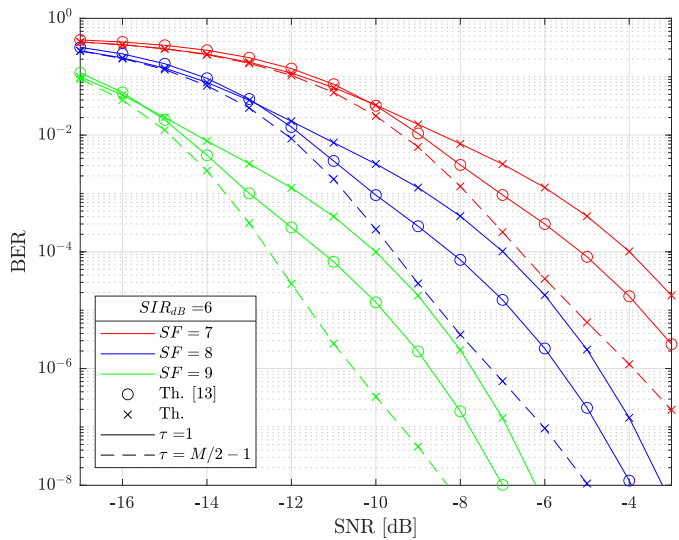}
  \caption{\textcolor{black}{Theoretical $BER$} comparison of LoRa interference for worst and best $\tau$ values ($\tau=1$ and $\tau=M/2-1$, respectively).
  $SF=\{7,8,9\}$ and $SIR_{dB}=6$. Theoretical average \textcolor{black}{Theoretical $BER$} of \cite{joerg5} is also reported.}
\label{fig:simu_LoRa_interf_compar_joerg_nous}
\end{figure}

The authors in \cite{joerg5} have already derived an approximated LoRa interference expression in aligned context (integer $\tau$).
They did not compute $SER$ performance for a specific $\tau$ but rather an average performance over $\tau \in \{0,1,\ldots,M/2\}$ with symmetric performance assumed (see proposition \ref{prop:sym}).
\textcolor{black}{This implies the need to modify \cite[eq.~(28)]{joerg5} for a given $\tau$.}
The interference only error probability \textcolor{black}{\cite[eq.~(28)]{joerg5}} becomes then $P_{e_{(I)}}^{(I)} \approx \frac{1}{M} \sum_{a_1=0}^{M-1} Q(\frac{1-\mathcal{U}_0(a_1,\tau)}{\sigma\sqrt{2}})$ for $\tau \in \{0,1,\ldots,M/2\}$.
$Q(.)$ denotes the Q-function and $\mathcal{U}_0$ given by \cite[eq.~(21)]{joerg5}.
Their final  \textcolor{black}{$SER$} expression is $P_e^{(I)} = P_{e_{(N)}}^{(I)} + (1 - P_{e_{(N)}}^{(I)}) P_{e_{(I)}}^{(I)}$ with $P_{e_{(N)}}^{(I)}$ the $AWGN$ only error probability derived in \cite[eq.~(21)]{joerg2}.

From Figure \ref{fig:simu_LoRa_interf}, theoretical performance of \cite{joerg5} at fixed $\tau$ agrees with our  \textcolor{black}{$SER$} except for $\tau$ around $M/2$.
The reason probably comes from the fine analysis of the \textcolor{black}{additive interference} terms in \eqref{eqdrenfg}-\eqref{eqdrenfgeq}
where performance loss could rise at some particular values of $\tau$ like $\tau=M/2$.
Their theoretical developments do not take into account such a behavior.
However, the computational complexity of  \textcolor{black}{$SER$} derived in \cite{joerg5} is very low in comparison to our results.
Nevertheless the computing time to evaluate our expressions remains reasonable. 

Figure \ref{fig:simu_LoRa_interf_compar_joerg_nous} highlights the extreme bounds of \textcolor{black}{Bit Error Rate ($BER$) $BER\approx SER/2$} performance reachable depending on $\tau$ for $SF \in \{7,8,9\}$.
The average \textcolor{black}{$BER$} over $\tau$ in \cite[Figure~3]{joerg5} is also reported in the figure.
It's worth noting that $\tau$ has a significant impact on performance when compared to the average performance over $\tau$.
As seen in Figure~\ref{fig:simu_LoRa_interf} best performance is obtained at $\tau = M/2-1$.
The less $\tau$ ($\tau<M/2$) the worse performance.
This can be seen as lower and upper performance bounds of results presented in \cite{joerg5}.
At low $SNR$ we observe a little bias for \cite{joerg5} because the average \textcolor{black}{$BER$} is slightly higher than the worst performance at $\tau=1$.
Simulation (not plotted for clarity) confirms the convergence of the lines at low $SNR$ for $\tau=1$ and $\tau=M/2-1$.

\section{Conclusion}

In this paper, a deep analysis of LoRa performance in $MPC$ is proposed.
An approximate closed-form of theoretical  \textcolor{black}{$SER$} for both coherent and non coherent detection schemes is derived and validated by simulation results.
This analysis highlights several significant LoRa behaviors under $MPC$.  
LoRa seems to be sensitive only to the first path of the exponential decay channel 
with reduced downside on performance for the other lower path gains.
This enables simple LoRa equalization schemes considering only a reduced number of channel paths (or the strongest path gain) to estimate.
For a path gain with $|\alpha|\ge0.4$, the performance degradation could be very \textcolor{black}{significant} in comparison to the $AWGN$ performance (e.g. performance loss about 12~dB for $|\alpha|=0.8$ whatever $SF$).
Otherwise, at fixed path gain, the less is the channel delay spread, the worse the performance.
This effect of delay spread vanishes for $|\alpha|\le0.4$.

\textcolor{black}{
As stated in the manuscript, performance evaluation is performed in the aligned context.
Our $SER$ results can be seen as pessimistic results in the sens that if time arrival of the echo is not multiple of the sampling rate, the echo energy is spread over neighbor $DFT$ bins that implies a lower parasitic peak energy in the most significant $DFT$ bin.
As we observe performance is dominated by the strongest $DFT$ bin (with the exponential decay channel), an equivalent performance evaluation could be deduced by considering only the most significant real peak amplitude in the $DFT$ bin.
If several parasitic peaks with equal (or near-equal) amplitudes appear in the $DFT$ bins (\emph{e.g.} time of arrival at half of the sample period), these values have to be considered in our analytical $SER$ expression to evaluate performance.
}

The article showed that theoretical developments for $MPC$ case can be \textcolor{black}{applied} to derive the
\textcolor{black}{aligned}
LoRa interference case, showing interesting results: the performance gap between  \textcolor{black}{$SER$} for $\tau$ around the half of symbol interval and for small $\tau$ value could be significant, and this difference is bigger as $SF$ increases.
\textcolor{black}{Performance degradation is extreme when the interference signal and the signal of interest arrive at the same time.
This extends in a complementary manner \cite{joerg5} findings by adding $SER$ deviation as a function of $\tau$ to the average performance.
}

We also show that the influence of the channel phase $\phi$ of the LoRa interferer vanishes on performance for the typical $SF\ge7$ values.
Only for \textcolor{black}{the} particular $\tau$ value at $M/2$ could exhibit  \textcolor{black}{$SER$} variation depending of~$\phi$.

\textcolor{black}{The authors from \cite{afisiadis2} took into account the non-aligned case for interference scenario and derived approximated $SER$ with performance averaged on $\tau$.}
\textcolor{black}{It may be interesting to extend our results by comparing \cite{afisiadis2} at fixed $\tau$ values and also consider $MPC$ study.}
\appendices

\section{Output $DFT$ in presence of $ISI$}
\label{sec:outDFT}
  
From (\ref{eqrak2_IES}), the $DFT$ of $\Tilde{r}_a[k]$ is equal to:
 \begin{equation}
 \label{eq:Ra1_demo}
   \Tilde{R}_a[n]=\alpha_0 \sum_{k=0}^{M-1} e^{2j\pi k \frac{a-n}{M}} +
   \sum_{i=1}^{K-1} S_i[n;\overline{a}] + \Tilde{W}[n]
 \end{equation}
 with:
 \begin{equation}
 \label{eq:Ra1_demo2}
   S_i[n;\overline{a}]=\Tilde{\alpha}_i(\overline{a})\sum_{k=0}^{M-1} e^{2j\pi k
     \frac{\overline{a} - k_i-n}{M}}
 \end{equation}
for $n=0,1,\ldots,M-1$.
The first term in (\ref{eq:Ra1_demo}) is equal to $M\delta[a-n]$ because of the following identity:

 \begin{equation}
 \label{eq:Ra1_demo3}
 \sum_{k=0}^{M-1} e^{2j\pi k \frac{q}{M}}=0\text{ for }q\in\mathbb{N}^*  
 \end{equation}

The second term $S_i[n;\overline{a}]$ in \eqref{eq:Ra1_demo} for $\overline{a}=a^-$ or $\overline{a}=a$ (depending of the time index $k$ in \eqref{eq:Ra1_demo2}) can be decomposed as:

 \begin{align}
   \label{eq:Mi_a_bar}
   S_i[n;\overline{a}]=& \Tilde{\alpha}_i(a^-)\sum_{k=0}^{k_i - 1} e^{2j\pi
     k/M(a^- - k_i - n)}\nonumber\\
 &+\Tilde{\alpha}_i(a)\sum_{k=k_i}^{M-1} e^{2j\pi
     k/M(a - k_i - n)}
 \end{align}

 Depending on the values of $n$, $S_i[n;\overline{a}]$ leads to the two following
 results:
 \begin{itemize}
 \item for $n\neq a-k_i$ and thanks to (\ref{eq:Ra1_demo3}), we obtain:
  \begin{align}
 S_i[n;\overline{a}]&=M_i[n;a^-]-\Tilde{\alpha}_i(a)\sum_{k=0}^{k_i - 1} e^{2j\pi
     k/M(a - k_i - n)}\\
 &=M_i[n;a^-]-M_i[n;a]
 \end{align}
 \item for $n=a-k_i$
 \begin{equation}
   \label{eq:Mi-a-k_i}
 S_i[n;\overline{a}]=M_i[n;a^-]+(M-k_i)\Tilde{\alpha}_i(a)  
 \end{equation}
 \end{itemize}
 
Equations \eqref{eq:Ra1_demo}-\eqref{eq:Mi-a-k_i} show the result summarized in \eqref{eqRak2_IES}.

\section*{Acknowledgment}

This work was jointly supported by the Brest Institute of Computer Science and Mathematics ($IBNM$) CyberIoT Chair of Excellence of the University of Brest, the Brittany Region and the “Pôle d’Excellence Cyber”.

\bibliographystyle{IEEEtran}
\bibliography{IEEEabrv,biblio}

\begin{IEEEbiography}[{\includegraphics[width=1in,height=1.25in,clip,keepaspectratio]{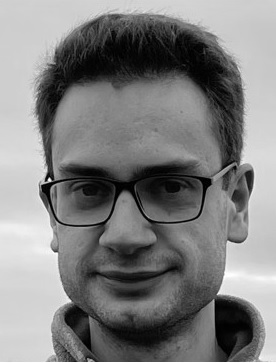}}]{Clément Demeslay} (M’20) received the M.Sc degree in Signal Processing field from University of Brest (UBO), France, in 2020.
As of October 2020, he is pursuing a Ph.D degree related with the development of intelligent jamming techniques to ensure secret, reliable and robust transmissions for several standards such as IoT (LoRa), 5G \& Beyond and also Space applications.
\end{IEEEbiography}

\begin{IEEEbiography}[{\includegraphics[width=1in,height=1.25in,clip,keepaspectratio]{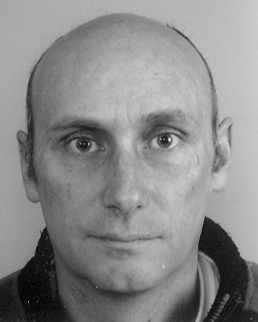}}]{Philippe Rostaing} (M'14) received the Ph.D degree in electrical engineering from the University of Nice-Sophia Antipolis, Nice, France, in 1997.
From 1997 to 2000, he was Assistant Professor with the French Naval Academy, Lanvéoc-Poulmic, France.
Since 2000, he has been Assistant Professor of digital communications and signal processing with the University of Brest, Brest, France. He is a member of the Laboratory Lab-STICC (UMR CNRS 6285) in the Security, Intelligence and Integrity of Information (SI3) Team.
His main research interests include signal processing and coding theory for wireless communications with emphasis on MIMO precoding systems.
\end{IEEEbiography}

\begin{IEEEbiography}[{\includegraphics[width=1in,height=1.25in,clip,keepaspectratio]{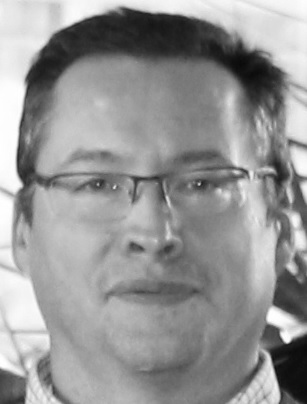}}]{Roland Gautier} (M'09) received the M.Sc degree from the University of Nice-Sophia Antipolis, France, in 1995, where his research activities were concerned with the blind source separation of convolutive mixtures for MIMO systems in digital communications.
He received the Ph.D degree in electrical engineering from the University of Nice-Sophia Antipolis, France, in 2000, where his research interests were in experiment design for nonlinear parameters models.
From 2000 to 2001, he was Assistant Professor with Polytech-Nantes, the engineering school of the University of Nantes, France.
Since September 2001, he has worked with the University of Brest, France, as an Associate Professor in electronic engineering and Signal Processing.
His general interests lie in the area of signal processing and digital communications.
His current research focuses on digital communication intelligence (COMINT), analysis, and blind parameters recognition, Multiple-Access and Spread Spectrum transmissions, Cognitive and Software Radio, Cybersecurity, Physical Layer Security for communications, Drones communications detection and Jamming.
From 2007 to June 2012, he was assistant manager of the Signal Processing Team, within the Laboratory for Science and Technologies of Information, Communication and Knowledge (Lab-STICC - UMR CNRS 6285).
From July 2012 to august 2020, he was the manager of the Defense Research Axis, within the COM Team of the Lab-STICC.
He received his Habilitation to Supervise Research (HDR) from the University of Brest in 2013 presenting an overview of his post-doctoral scientific research activities on the development of "self-configuring multi-standard adaptive receivers : Blind analysis of digital transmissions for military communications and Cognitive Radio".
Since 2018, he is the holder of the "CyberIoT" Chair of Excellence.
Since September 2020, he is the manager of the Security, Intelligence and Integrity of Information (SI3) Team of the Lab-STICC.
\end{IEEEbiography}

\end{document}

%% file: figures/demes1.tex
\tikzset{every picture/.style={line width=0.75pt}} 

\begin{tikzpicture}[x=0.75pt,y=0.75pt,yscale=-0.8,xscale=0.8]

\draw    (95.5,36) -- (425.5,36) ;
\draw    (95.5,65) -- (425.5,65) ;
\draw    (130,50) -- (130,80) ;
\draw    (240,20) -- (240,40) ;
\draw  [color={rgb, 255:red, 0; green, 0; blue, 255 }  ,draw opacity=1 ][fill={rgb, 255:red, 0; green, 0; blue, 255 }  ,fill opacity=1 ] (240,30) -- (370,30) -- (370,40) -- (240,40) -- cycle ;
\draw  [color={rgb, 255:red, 0; green, 0; blue, 255 }  ,draw opacity=1 ][fill={rgb, 255:red, 0; green, 0; blue, 255 }  ,fill opacity=1 ] (260,60) -- (390,60) -- (390,70) -- (260,70) -- cycle ;
\draw  [color={rgb, 255:red, 255; green, 0; blue, 0 }  ,draw opacity=1 ][fill={rgb, 255:red, 255; green, 0; blue, 0 }  ,fill opacity=1 ] (110,30) -- (240,30) -- (240,40) -- (110,40) -- cycle ;
\draw  [color={rgb, 255:red, 255; green, 0; blue, 0 }  ,draw opacity=1 ][fill={rgb, 255:red, 255; green, 0; blue, 0 }  ,fill opacity=1 ] (130,60) -- (260,60) -- (260,70) -- (130,70) -- cycle ;
\draw    (260,50) -- (260,80) ;
\draw    (390,50) -- (390,80) ;
\draw    (110,20) -- (110,50) ;
\draw    (370,20) -- (370,50) ;

\draw (6.5,25.4) node [anchor=north west][inner sep=0.75pt]    {$\alpha _{0} s[ k]$};
\draw (6.5,55.4) node [anchor=north west][inner sep=0.75pt]    {$\alpha _{1} s[ k-k_{1}]$};
\draw (251,80.4) node [anchor=north west][inner sep=0.75pt]    {$k_{1}$};
\draw (237,2.4) node [anchor=north west][inner sep=0.75pt]    {$0$};
\draw (346.5,7.4) node [anchor=north west][inner sep=0.75pt]    {$M-1$};
\draw (290,11.9) node [anchor=north west][inner sep=0.75pt]  [color={rgb, 255:red, 0; green, 0; blue, 255 }  ,opacity=1 ]  {$a$};
\draw (170,11.4) node [anchor=north west][inner sep=0.75pt]  [color={rgb, 255:red, 255; green, 0; blue, 0 }  ,opacity=1 ]  {$a^{-}$};
\draw (245,99) node [anchor=north west][inner sep=0.75pt]   [align=left] {detection of the\\[-.4em]current symbol $\displaystyle a$};
\draw (311,42.4) node [anchor=north west][inner sep=0.75pt]  [color={rgb, 255:red, 0; green, 0; blue, 255 }  ,opacity=1 ]  {$a$};
\draw (190,41.4) node [anchor=north west][inner sep=0.75pt]  [color={rgb, 255:red, 255; green, 0; blue, 0 }  ,opacity=1 ]  {$a^{-}$};
\draw  [dash pattern={on 4.5pt off 4.5pt}]  (240,30) -- (240,135) ;
\draw  [dash pattern={on 4.5pt off 4.5pt}]  (370,30) -- (370,135) ;
\end{tikzpicture}

%% file: figures/demes4.tex
  \begin{tikzpicture}[x=0.75pt,y=0.75pt,yscale=-0.8,xscale=0.83]
        
        \draw   (122,79.67) -- (211.07,79.67) -- (211.07,111.16) -- (122,111.16) -- cycle ;
        \draw   (211.07,79.67) -- (300.14,79.67) -- (300.14,111.16) -- (211.07,111.16) -- cycle ;
        \draw   (300.14,79.67) -- (389.21,79.67) -- (389.21,111.16) -- (300.14,111.16) -- cycle ;
        \draw   (160,111) -- (249.07,111) -- (249.07,142.49) -- (160,142.49) -- cycle ;
        \draw   (249.07,111) -- (338.14,111) -- (338.14,142.49) -- (249.07,142.49) -- cycle ;
        \draw   (338.14,111) -- (427.21,111) -- (427.21,142.49) -- (338.14,142.49) -- cycle ;
        \draw [color={rgb, 255:red, 255; green, 0; blue, 0 }  ,draw opacity=1 ]   (211.07,79.67) -- (210.4,141.82) ;
        \draw [color={rgb, 255:red, 255; green, 0; blue, 0 }  ,draw opacity=1 ]   (300.47,80.08) -- (299.81,142.23) ;
        \draw   (122,128.74) -- (131.5,121) -- (131.5,124.87) -- (150.5,124.87) -- (150.5,121) -- (160,128.74) -- (150.5,136.49) -- (150.5,132.62) -- (131.5,132.62) -- (131.5,136.49) -- cycle ;
           
        \draw (248.67,88.73) node [anchor=north west][inner sep=0.75pt]    {$a_1$};
        \draw (275.67,121.07) node [anchor=north west][inner sep=0.75pt]    {$a_2$};
        \draw (163,84.4) node [anchor=north west][inner sep=0.75pt]    {$a^{-}_{1}$};
        \draw (184.33,114.4) node [anchor=north west][inner sep=0.75pt]    {$a^{-}_{2}$};
        \draw (338.33,82.4) node [anchor=north west][inner sep=0.75pt]    {$a^{+}_{1}$};
        \draw (364.33,115.07) node [anchor=north west][inner sep=0.75pt]    {$a^{+}_{2}$};
        \draw (135.33,137.07) node [anchor=north west][inner sep=0.75pt]    {$\tau$};
        \draw (52,84.73) node [anchor=north west][inner sep=0.75pt]    {$\alpha_0 s[k]$};
        \draw (26,118.73) node [anchor=north west][inner sep=0.75pt]    {$\alpha_1 s^{(I)}[ k-\tau ]$};
    \end{tikzpicture}

%% file: figures/demes5.tex
    \begin{tikzpicture}[x=0.75pt,y=0.75pt,yscale=-0.8,xscale=0.52]
        \draw    (40.67,220) -- (627.74,220.49) ;
        \draw [shift={(629.74,220.49)}, rotate = 180.05] [color={rgb, 255:red, 0; green, 0; blue, 0 }  ][line width=0.75]    (10.93,-3.29) .. controls (6.95,-1.4) and (3.31,-0.3) .. (0,0) .. controls (3.31,0.3) and (6.95,1.4) .. (10.93,3.29)   ;
        \draw    (305.33,219.67) -- (305.07,82.49) ;
        \draw [shift={(305.07,80.49)}, rotate = 449.89] [color={rgb, 255:red, 0; green, 0; blue, 0 }  ][line width=0.75]    (10.93,-3.29) .. controls (6.95,-1.4) and (3.31,-0.3) .. (0,0) .. controls (3.31,0.3) and (6.95,1.4) .. (10.93,3.29)   ;
        \draw    (189.33,220) -- (189.08,130.49) ;
        \draw [shift={(189.07,128.49)}, rotate = 449.84] [color={rgb, 255:red, 0; green, 0; blue, 0 }  ][line width=0.75]    (10.93,-3.29) .. controls (6.95,-1.4) and (3.31,-0.3) .. (0,0) .. controls (3.31,0.3) and (6.95,1.4) .. (10.93,3.29)   ;
        \draw    (430,220.67) -- (429.74,131.16) ;
        \draw [shift={(429.74,129.16)}, rotate = 449.84] [color={rgb, 255:red, 0; green, 0; blue, 0 }  ][line width=0.75]    (10.93,-3.29) .. controls (6.95,-1.4) and (3.31,-0.3) .. (0,0) .. controls (3.31,0.3) and (6.95,1.4) .. (10.93,3.29)   ;
        \draw    (40.67,220) -- (40.41,30.49) ;
        \draw [shift={(40.4,28.49)}, rotate = 449.92] [color={rgb, 255:red, 0; green, 0; blue, 0 }  ][line width=0.75]    (10.93,-3.29) .. controls (6.95,-1.4) and (3.31,-0.3) .. (0,0) .. controls (3.31,0.3) and (6.95,1.4) .. (10.93,3.29)   ;
        \draw [color={rgb, 255:red, 255; green, 0; blue, 0 }  ,draw opacity=1 ]   (186.67,239.67) .. controls (183.78,272.65) and (258.12,255.04) .. (297.94,239.21) ;
        \draw [shift={(299.74,238.49)}, rotate = 517.86] [color={rgb, 255:red, 255; green, 0; blue, 0 }  ,draw opacity=1 ][line width=0.75]    (10.93,-3.29) .. controls (6.95,-1.4) and (3.31,-0.3) .. (0,0) .. controls (3.31,0.3) and (6.95,1.4) .. (10.93,3.29)   ;
        \draw [color={rgb, 255:red, 0; green, 255; blue, 0 }  ,draw opacity=1 ]   (435.11,236.24) .. controls (432.21,269.39) and (369.38,274.15) .. (315.73,240.28) ;
        \draw [shift={(314.11,239.24)}, rotate = 392.95] [color={rgb, 255:red, 0; green, 255; blue, 0 }  ,draw opacity=1 ][line width=0.75]    (10.93,-3.29) .. controls (6.95,-1.4) and (3.31,-0.3) .. (0,0) .. controls (3.31,0.3) and (6.95,1.4) .. (10.93,3.29)   ;
         
        \draw (616.67,223.07) node [anchor=north west][inner sep=0.75pt]    {$n$};
        \draw (299.33,222.07) node [anchor=north west][inner sep=0.75pt]    {$a_{1}$};
        \draw (146.67,218.73) node [anchor=north west][inner sep=0.75pt]    {$a_{2} -\tau $};
        \draw (386.67,219.07) node [anchor=north west][inner sep=0.75pt]    {$a^{-}_{2} -\tau $};
        \draw (290.33,52.4) node [anchor=north west][inner sep=0.75pt]    {$M$};
        \draw (94.67,105.07) node [anchor=north west][inner sep=0.75pt]    {$\Tilde{\alpha} _{\tau }( a_{2})( M-\tau )$};
        \draw (380.67,97.07) node [anchor=north west][inner sep=0.75pt]    {$\Tilde{\alpha} _{\tau }\left( a^{-}_{2}\right) \tau$};
        \draw (2,0) node [anchor=north west][inner sep=0.75pt]    {$\tilde{R}[n]$};
    \end{tikzpicture}  

%% file: figures/demes6.tex
  \begin{tikzpicture}[x=0.75pt,y=0.75pt,yscale=-0.8,xscale=0.52]
        
        \draw    (40.67,220) -- (627.74,220.49) ;
        \draw [shift={(629.74,220.49)}, rotate = 180.05] [color={rgb, 255:red, 0; green, 0; blue, 0 }  ][line width=0.75]    (10.93,-3.29) .. controls (6.95,-1.4) and (3.31,-0.3) .. (0,0) .. controls (3.31,0.3) and (6.95,1.4) .. (10.93,3.29)   ;
        \draw    (305.33,219.67) -- (305.07,82.49) ;
        \draw [shift={(305.07,80.49)}, rotate = 449.89] [color={rgb, 255:red, 0; green, 0; blue, 0 }  ][line width=0.75]    (10.93,-3.29) .. controls (6.95,-1.4) and (3.31,-0.3) .. (0,0) .. controls (3.31,0.3) and (6.95,1.4) .. (10.93,3.29)   ;
        \draw    (189.33,220) -- (189.08,130.49) ;
        \draw [shift={(189.07,128.49)}, rotate = 449.84] [color={rgb, 255:red, 0; green, 0; blue, 0 }  ][line width=0.75]    (10.93,-3.29) .. controls (6.95,-1.4) and (3.31,-0.3) .. (0,0) .. controls (3.31,0.3) and (6.95,1.4) .. (10.93,3.29)   ;
        \draw    (40.67,220) -- (40.41,30.49) ;
        \draw [shift={(40.4,28.49)}, rotate = 449.92] [color={rgb, 255:red, 0; green, 0; blue, 0 }  ][line width=0.75]    (10.93,-3.29) .. controls (6.95,-1.4) and (3.31,-0.3) .. (0,0) .. controls (3.31,0.3) and (6.95,1.4) .. (10.93,3.29)   ;
        \draw [color={rgb, 255:red, 0; green, 0; blue, 255 }  ,draw opacity=1 ]   (186.67,239.67) .. controls (183.78,272.65) and (258.12,255.04) .. (297.94,239.21) ;
        \draw [shift={(299.74,238.49)}, rotate = 517.86] [color={rgb, 255:red, 0; green, 0; blue, 255 }  ,draw opacity=1 ][line width=0.75]    (10.93,-3.29) .. controls (6.95,-1.4) and (3.31,-0.3) .. (0,0) .. controls (3.31,0.3) and (6.95,1.4) .. (10.93,3.29)   ;
        
        \draw (616.67,223.07) node [anchor=north west][inner sep=0.75pt]    {$n$};
        \draw (293.33,222.07) node [anchor=north west][inner sep=0.75pt]    {$a_{1}$};
        \draw (146.67,218.73) node [anchor=north west][inner sep=0.75pt]    {$a_{2} -\tau $};
        \draw (289.33,55.4) node [anchor=north west][inner sep=0.75pt]    {$M$};
        \draw (134.67,102.73) node [anchor=north west][inner sep=0.75pt]    {$\Tilde{\alpha} _{\tau }( a_{2}) M$};
        \draw (2,0) node [anchor=north west][inner sep=0.75pt]    {$\tilde{R}[ n]$};
    \end{tikzpicture}